\documentclass[a4paper,UKenglish,cleveref, autoref, thm-restate]{lipics-v2021}
\nolinenumbers
\usepackage{amsmath}
\usepackage{amssymb}
\usepackage{amsthm}
\usepackage{dsfont}
\usepackage{color}
\usepackage{amsmath}
\usepackage{amsfonts}
\usepackage{amssymb}
\usepackage{bbm}
\usepackage{relsize}
\usepackage[procnumbered, linesnumbered,noend,algoruled]{algorithm2e}
\usepackage{graphicx}
\graphicspath{ {./images/} }
\usepackage{tikz}
\usetikzlibrary{arrows}
\tikzset{
	semithick,
	node distance = 2cm,
	dot/.style={circle,fill,inner sep=2pt}
}
\tikzset{
	side by side/.style 2 args={
		line width=2pt,
		#1,
		postaction={
			clip,postaction={draw,#2}
		}
	}
}
\tikzstyle{every state}=[draw = black,thick,fill = white,minimum size = 4mm]
\tikzstyle{selected edge} = [draw,line width=2pt,-,red!50]
\tikzset{
	edge/.style={->,> = latex'}
}

\newcommand{\negA}{\vspace{-0.05in}}

\newcommand{\ceil}[1]{{\left\lceil#1  \right\rceil}}

\newcommand{\cA}{{\mathcal{A}}}

\newcommand{\cG}{{\mathcal{G}}}
\newcommand{\cI}{{\mathcal{I}}}
\newcommand{\cm}{{\mathcal{M}}}

\newcommand{\ck}{{\mathcal{K}}}

\newcommand{\cC}{{\mathcal{C}}}

\newcommand{\OPT}{\textnormal{OPT}}

\newcommand{\eps}{{\varepsilon}}

\newcommand{\floor}[1]{\left\lfloor #1 \right\rfloor}

\DeclareMathOperator*{\argmin}{arg\,min}

\title{An EPTAS for Budgeted Matching and Budgeted Matroid Intersection via Representative Sets} 

\titlerunning{An EPTAS for Budgeted Matching and Budgeted Matroid Intersection} 
\author{ Ilan Doron-Arad}{ Computer Science Department, 
		Technion, Haifa, Israel.}{idoron-arad@cs.technion.ac.il}{}{}

\author{ Ariel Kulik}{CISPA Helmholtz Center for Information Security, Saarland Informatics Campus, Germany.}{ariel.kulik@cispa.de}{}{}

\author{Hadas Shachnai}{Computer Science Department, 
		Technion, Haifa, Israel.}{hadas@cs.technion.ac.il}{}{}
\authorrunning{ } 

\Copyright{ } 



\relatedversion{} 

\EventEditors{}
\EventNoEds{0}
\EventLongTitle{}
\EventShortTitle{}
\EventAcronym{}
\EventYear{}
\EventDate{}
\EventLocation{}
\EventLogo{}
\SeriesVolume{42}
\ArticleNo{23}

\begin{document}

\maketitle

\begin{abstract}
We study the budgeted versions of the well known matching and matroid intersection problems. While both problems admit a {\em polynomial-time approximation scheme (PTAS)}	[Berger et al. (Math. Programming, 2011), Chekuri, Vondr{\'a}k and Zenklusen (SODA 2011)],
it has been an intriguing open question whether these problems admit a {\em fully} PTAS (FPTAS), or even an {\em efficient} PTAS (EPTAS).  
	
In this paper we answer the 
second part of this question affirmatively, 
by presenting an EPTAS for budgeted matching and  
budgeted matroid intersection. A main component of our scheme is a novel construction of {\em representative sets} for desired solutions, whose cardinality depends only on $\eps$, the accuracy parameter. Thus, enumerating over solutions within a representative set leads to an EPTAS. This crucially distinguishes our algorithms from previous approaches, which rely on {\em exhaustive} enumeration over the solution set. Our ideas for constructing representative sets 
may find use in	tackling other budgeted optimization problems, and are thus
of independent interest.
\end{abstract}

\section{Introduction}
\label{sec:introduction}

A wide range of NP-hard combinatorial optimization problems can be formulated as follows. We are given a ground set $E$ and a family $\cm$ of subsets of $E$ called the {\em feasible sets}. 
The elements in the ground set are associated with 
a cost function $c:E\rightarrow \mathbb{R}_{\geq 0}$ and a profit function $p:E\rightarrow \mathbb{R}$, and we are also given a budget $\beta \in \mathbb{R}_{\geq 0}$. A {\em solution} is a feasible set $S \in \cm$ of bounded cost $c(S) \leq \beta$.\footnote{For a function $f:A \rightarrow \mathbb{R}$ and a subset of elements $C \subseteq A$, we define $f(C) = \sum_{e \in C} f(e)$.} Generally, the goal is to find a solution $S$ of maximum profit, that is: 
\begin{equation}
\label{eq:1}
\max p(S) \text{ s.t. } S \in \cm, c(S) \leq \beta.
\end{equation} 
Notable examples include 
shortest weight-constrained path \cite{garey1979computers}, constrained minimum spanning trees \cite{ravi1996constrained}, and knapsack with a conflict graph  \cite{pferschy2009knapsack}. In this work, we focus on two prominent problems which can be formulated as~\eqref{eq:1}. 

In  the {\em budgeted matching (BM)} problem we are given an undirected graph $G = (V,E)$,  profit and cost functions on the edges $p,c:E \rightarrow \mathbb{R}_{\geq 0}$, and a budget $\beta \in \mathbb{R}_{\geq 0}$. A {\em solution} is 
a {\em matching} $S \subseteq E$ in $G$ such that $c(S) \leq \beta$. The goal is to find a solution $S$ such that the total profit $p(S)$ is maximized. Observe that BM can be formulated using \eqref{eq:1}, by letting
$\cm$ be the set of matchings in $G$.

In  the {\em budgeted matroid intersection (BI)} problem we are given two matroids $(E,\cI_1)$ and $(E,\cI_2)$ over a ground set $E$,  profit and cost functions on the elements $p,c:E \rightarrow \mathbb{R}_{\geq 0}$, and a budget $\beta \in \mathbb{R}_{\geq 0}$. Each matroid is given by a membership oracle. 
A {\em solution} is a {\em common independent set} $S \in \cI_1 \cap \cI_2$ such that $c(S) \leq \beta$; the goal is to find a solution $S$ of maximum total profit $p(S)$. The formulation of BI as \eqref{eq:1} follows by defining the feasible sets as all common independent sets $\cm =  \cI_1 \cap \cI_2$.

Let $\OPT(I)$ be the value of an optimal solution for an instance $I$ of a maximization problem $\Pi$. For $\alpha \in (0,1]$, we say that $\cA$ is an $\alpha$-approximation algorithm for $\Pi$ if, for any instance $I$ of $\Pi$,
$\cA$ outputs a solution of value at least $\alpha \cdot \OPT(I)$. A {\em polynomial-time approximation scheme} (PTAS)
for $\Pi$ is a family of algorithms $(A_{\eps})_{\eps>0}$ such that, for any $\eps>0$, $A_{\eps}$ is a polynomial-time $(1 - \eps)$-approximation algorithm for $\Pi$. 
As $\eps$ gets smaller, a running time of the form $n^{\Theta\left(\frac{1}{\eps}\right)}$ for a PTAS 
may become prohibitively large 
and thus impractical; therefore, it is natural to seek approximation schemes with better running times. Two families of such schemes 
have been extensively studied:
an {\em efficient PTAS} (EPTAS) is a PTAS $(A_{\eps})_{\eps>0}$ whose running time is of the form $f\left(\frac{1}{\eps}\right) \cdot n^{O(1)}$, where $f$ is an arbitrary computable function, and $n$ is the bit-length encoding size
of the input instance. In a {\em fully PTAS} (FPTAS) the running time of $A_{\eps}$ is of the form $ {\left(\frac{n}{\eps}\right)}^{O(1)}$. For comprehensive surveys on approximation schemes see, e.g.,~\cite{SW01,AAM_2ndEd}.

The state of the art for BM and BI is a PTAS developed by Berger et al. \cite{BBGS11}. Similar results for both problems 
follow from a later work of Chekuri et al. \cite{chekuri2011multi} for the multi-budgeted variants of BM and BI. The running times of the above  
schemes are dominated by exhaustive enumeration which finds a set of $\Theta\left( \frac{1}{\eps} \right)$ elements of highest profits in the solution.
In this paper we optimize the enumeration procedure using a novel approach, 
which enables to substantially reduce the size of the domain over which we seek an efficient solution. Our main results are the following.
\begin{theorem}
\label{thm:matching}
There is an \textnormal{EPTAS} for the budgeted matching problem. 
\end{theorem}

\begin{theorem}
\label{thm:matroid}
There is an \textnormal{EPTAS} for the budgeted matroid intersection problem. 
\end{theorem} 

\subsection{Related Work}

  BM and BI are immediate generalizations of the classic $0/1$-knapsack problem. While the knapsack problem is known to be NP-hard, it admits an FPTAS. This raises a natural question whether BM and BI admit an FPTAS as well. 
  The papers~\cite{BBGS11,chekuri2011multi} along with our results can be viewed as
  first steps towards answering this question.
  
  Berger et al. \cite{BBGS11} developed the first PTAS for BM and BI. Their approach includes an elegant combinatorial algorithm for {\em patching} two solutions for the {\em Lagrangian relaxation} of the underlying problem (i.e., BM or BI); one solution is feasible but has small profit, while the other solution has high profit but is infeasible. The scheme of~\cite{BBGS11} enumerates over solutions containing only high profit elements and uses the combinatorial algorithm to add low profit elements. This process may result in losing (twice) the profit of a low profit element, leading to a PTAS.

Chekuri et al.~\cite{chekuri2011multi} developed a PTAS for multi-budgeted matching and a randomized PTAS for multi-budgeted matroid intersection; these are variants of BM and BI, respectively, in which the costs are $d$-dimensional, for some constant $d \geq 2$. They incorporate a non-trivial martingale based analysis to derive the results, along with enumeration to facilitate the selection of profitable elements for the solution. The paper~\cite{chekuri2011multi} generalizes a previous result of Grandoni and Zenklusen \cite{grandoni2010approximation}, who obtained a PTAS for multi-budgeted matching and multi-budgeted matroid intersection in {\em representable matroids}.\footnote{Representable matroids are also known as {\em linear matroids}.} For $d \geq 2$, the multi-budgeted variants of BM and BI generalize the two-dimensional knapsack problem, and thus do not admit an EPTAS unless W[1] = FPT \cite{kulik2010there}.

An evidence for the difficulty of attaining an FPTAS for BM comes from the {\em exact} variant of the problem. In this setting, we are given a graph $G = (V,E)$, a cost function $c: E \rightarrow \mathbb{R}_{\geq 0}$, and a {\em target} $B \in \mathbb{R}_{\geq 0}$; the goal is 
to find a perfect matching $S \subseteq E$ with exact specified cost $c(S) = B$. There is a randomized pseudo-polynomial time algorithm for exact matching \cite{mulmuley1987matching}. On the other hand, it is a long standing open question whether exact matching admits a 
deterministic pseudo-polynomial time algorithm \cite{papadimitriou1982complexity}. 
Interestingly, as noted by Berger et al. \cite{BBGS11}, an FPTAS for BM would give an affirmative answer also for the latter question. 
An FPTAS for BI would have similar implications for the {\em exact} matroid intersection problem, which admits a randomized (but not a deterministic) 
pseudo-polynomial time algorithm~\cite{camerini1992random}. While the above does not rule out the existence of an FPTAS for BM or BI, it indicates that improving our results from EPTAS to FPTAS might be a difficult task. 

For the budgeted matroid independent set (i.e., the special case of BI of two identical matroids), Doron-Arad et al.~\cite{DKS23} developed an EPTAS using {\em representative sets} to enhance enumeration over elements of high profits.\footnote{We elaborate below on the framework of~\cite{DKS23} vs. our notion of representative sets.} Their scheme 
exploits integrality properties of matroid polytopes under budget constraints (introduced in \cite{grandoni2010approximation}) to efficiently combine elements of low profit into the solution.

\subsection{Contribution and Techniques}

Given an instance $I$ of BM or BI, we say that an element $e$ is {\em profitable} if $p(e) > \eps \cdot \OPT(I)$; otherwise, $e$ is {\em non-profitable}.
 The scheme for BM and BI of Berger et al. \cite{BBGS11} distinguishes between profitable and non-profitable elements. 
 In the main loop, the algorithm enumerates over all potential solutions containing only profitable elements.\footnote{A similar technique is used also by Chekuri et al. \cite{chekuri2011multi}.} Each solution is extended
 to include non-profitable elements using a combinatorial algorithm.
 The algorithm outputs a solution of highest profit.
  Overall, this process may lose at most twice the profit of a non-profitable element in comparison to the optimum, which effectively preserves the approximation guarantee; however, an exhaustive enumeration over the profitable elements renders the running time $n^{\Omega\left(\frac{1}{\eps}\right)}$. In stark contrast, in this paper we introduce a new approach which enhances the enumeration over profitable elements, leading to an EPTAS.

We restrict the enumeration to only a small subset of elements called {\em representative set}; that is, a subset of elements $R \subseteq E$ satisfying the following property: there is a solution $S$ such that the profitable elements in $S$ are a subset of $R$, and the profit of $S$ is at least $(1-O(\eps)) \cdot \OPT(I)$. If one finds efficiently a representative set $R$ of cardinality $|R| \leq f\left( \frac{1}{\eps}\right)$ for some computable function $f$, obtaining an EPTAS is straightforward based on the approach of \cite{BBGS11}. 

Our scheme generalizes the {\em representative set} framework of Doron-Arad 
et al.~\cite{DKS23}, developed originally for budgeted matroid independent set. They construct a representative set as a basis of minimum cost of a matroid, which can be implemented using a greedy algorithm. Alas, a greedy analogue for the setting of matching and matroid intersection fails; we give an example in Figure~\ref{pic:example}.\footnote{The example is made clear once the reader is familiar with the definitions given in Section~\ref{sec:alg}.} Hence, we take a different approach. Our main technical contribution is in the novel construction of representative sets for each of our problems.

For BM we design a surprisingly simple algorithm which finds a representative set using a union of multiple matchings. 
To this end, we partition the edges in $G$ into {\em profit classes} such that each profit class contains edges of {\em similar} profits. We then use the greedy approach to repeatedly find in each profit class a union of disjoint matchings, where each matching has a bounded cardinality and is greedily selected to minimize cost. Intuitively, to show that the above yields a representative set, consider a profitable edge $e$ in some optimal solution. Suppose that $e$ is not chosen to our union of matchings, then we consider two cases. If each matching selected in the profit class of $e$ contains an edge that is adjacent to (i.e., shares a vertex with) $e$, we show that at least one of these edges can be exchanged with $e$;
otherwise, there exists a matching with no edge adjacent to $e$. In this case, 
we show that our greedy selection guarantees the existence of an edge in this matching which can be exchanged with $e$, implying
the above is a representative set (see the details in Section~\ref{sec:lemMainProofMatching}).

For BI, we design a recursive algorithm that relies on an {\em asymmetric interpretation} of the two given matroids. In each recursive call, we are given an independent set $S \in \cI_{1}$.
The algorithm adds to the constructed representative set a minimum cost basis $B_S$ of the second matroid $(E,\cI_{2})$, with the crucial restriction   
that any element $e \in B_S$ must satisfy $S \cup \{e\} \in \cI_1$.
Succeeding recursive calls will then use the set $S \cup \{e\}$, for every $e \in B_S$.
 Thus, we limit the search space to $\cI_1$, while bases are constructed w.r.t. $\cI_2$. To show that the algorithm yields a representative set,  consider a profitable element $f$ in an optimal solution. 
 We construct a sequence of elements which are independent w.r.t. $\cI_{1}$ and can be exchanged with $f$ w.r.t. $\cI_{2}$. Using matroid properties we show that one of these elements can be exchanged with $f$ w.r.t. both matroids (see the details in Section ~\ref{sec:lemMainProof}).

Interestingly, our framework for solving BM and BI (presented in Section~\ref{sec:alg}) can be extended to solve other problems formulated as~\eqref{eq:1} which possess similar {\em exchange properties}. We elaborate on that in Section~\ref{sec:discussion}. Moreover, our algorithms for constructing  representative sets for BM and BI may find use in tackling
other budgeted optimization problems (see Section~\ref{sec:discussion}), and are thus of independent interest.

\noindent
{\bf Organization of the paper:} In Section~\ref{sec:preliminaries}  we give some definitions and notation. Section~\ref{sec:alg} presents our framework that yields an EPTAS for each of the problems. In Sections~\ref{sec:lemMainProofMatching} and~\ref{sec:lemMainProof} we describe the algorithms for constructing representative sets for BM and BI, respectively. We conclude in Section~\ref{sec:discussion} with a summary and some directions for future work.
Due to space constraints, some of the proofs are given in the Appendix. 

\section{Preliminaries}
\label{sec:preliminaries}

For simplicity of the notation, for any set $A$ and an element $e$,  we use $A+e$ and  $A-e$ to denote $A \cup \{e\}$ and $A \setminus \{e\}$, respectively. Also, for any $k \in \mathbb{R}$ let $[k] = \{1,2,\ldots,\floor{k}\}$. 
For a function $f:A \rightarrow \mathbb{R}_{\geq 0}$ and a subset of elements $C \subseteq A$, let $f|_C:C \rightarrow \mathbb{R}_{\geq 0}$ be the {\em restriction} of $f$ to $C$, such that  $\forall e \in C: f|_C(e) = f(e)$. 

\subsection{Matching and Matroids}
 Given an undirected graph $G = (V,E)$, a {\em matching} of $G$ is a subset of edges $M \subseteq E$ such that each vertex appears as an endpoint in at most one edge in $M$, i.e., for all $v \in V$ it holds that $|\{\{u,v\} \in M~|~u \in V\}| \leq 1$. We denote by $V(M) = \{v \in V~|~\exists u \in V \text{ s.t. } \{u,v\} \in M\}$  the set of endpoints of a matching $M$ of $G$.
 
Let $E$ be a finite ground set and $\cI \subseteq 2^E$ a non-empty set containing subsets of $E$ called the {\em independent sets} of $E$. Then $\cm = (E, \cI)$ is a {\em matroid} if the following hold. 

\begin{enumerate}
\item (Hereditary Property) For all $A \in \cI$ and $B \subseteq A$, it holds that $B \in \cI$.
	
\item(Exchange Property) For any $A,B \in \cI$ where $|A| > |B|$, there is $e \in A \setminus B$ such that $B +e \in \cI$. 
\end{enumerate}

A {\em basis} of a matroid $\cG = (E, \cI)$ is an independent set $B \in \cI$ such that for all $e \in E \setminus B$ it holds that $B+e \notin \cI$. Given a cost function $c:E \rightarrow \mathbb{R}_{\geq 0}$, we say that a basis $B$ of $\cG$ is a {\em minimum} basis of $\cG$ w.r.t. $c$ if, for any basis $A$ of $\cG$ it holds that $c(B) \leq c(A)$. A minimum basis of $\cG$ w.r.t. $c$ can be easily constructed in polynomial-time using a greedy approach (see, e.g.,~\cite{cormen2022introduction}). In the following we define several matroid operations. Note that the structures resulting from the operations outlined in 
Definition~\ref{def:matroids} are matroids.
(see, e.g., \cite{schrijver2003combinatorial}). 

\begin{definition}
	\label{def:matroids}
	Let $\cG = (E, \cI)$ be a matroid.  
	\begin{enumerate}
		
		\item (restriction) For any $F \subseteq E$ define $\cI_{\cap F} = \{A \in \cI~|~ A \subseteq F\}$ and $\cG \cap F = (F, \cI_{\cap F})$.\label{prop1:restriction}
		
		\item (thinning) For any $F \in \cI$ define $\cI / F = \{A \subseteq E \setminus F~|~ A \cup F \in \cI\}$  and $\cG / F = (E \setminus F, \cI / F)$.\footnote{Thinning is generally known as contraction; we use the term thinning to avoid confusion with edge contraction in graphs.}\label{prop1:contraction}
		
		\item (truncation) For any $q \in \mathbb{N}$ define $\cI_{\leq q} = \{A \in \cI~|~ |A| \leq q\}$ and $[\cG]_{\leq q} = (E, \cI_{\leq q})$.\label{prop1:trunc}
	
	\end{enumerate}
\end{definition}

\subsection{Instance Definition}

We give a unified definition for instances of budgeted matching and budgeted matroid intersection. Given a ground set $E$ of elements, we say that $\cC$ is a {\em constraint} of $E$ if one of the following holds. 
\begin{itemize}
		\item $\cC = (V,E)$ is a {\em matching constraint}, where $\cC$ is an undirected graph. 
	Let $\cm(\cC) = \{M \subseteq E~|~M \text{ is a matching in $\cC$}\}$ be the {\em feasible sets} of $\cC$. 
	Given a subset of edges $F \subseteq E$, let $E / F = \{\{u,v\} \in E~|~u,v \notin V(F)\}$ be the {\em thinning} of $F$ on $E$, and let $\cC / F = (V, E / F)$ be the {\em thinning} of $F$ on $\cC$.  
	
	\item $\cC = (\cI_1,\cI_2)$ is a {\em matroid intersection constraint}, where $(E,\cI_1)$ and $(E,\cI_2)$ are matroids.  Throughout this paper, we assume that each of the matroids is given by an independence oracle. That is, determining whether some $F \subseteq E$ belongs to $\cI_1$ or to $\cI_2$ requires a single call to the corresponding oracle of $\cI_1$ or $\cI_2$, respectively. Let $\cm(\cC) =  \cI_1 \cap \cI_2$ be the collection of {\em feasible sets} of $\cC$. 
	 In addition, given some $F \subseteq E$, let $\cC / F = (\cI_1 / F, \cI_2 / F)$ be the {\em thinning} of $F$ on $\cC$. 
	 We say that $\cC$ is a {\em single matroid constraint} if $\cI_1 = \cI_2$

\end{itemize}
When understood from the context, we simply use $\cm = \cm(\cC)$. Define an instance of the {\em budgeted constrained (BC)} problem as a tuple $I = (E, \cC, c,p, \beta)$, where $E$ is a ground set of elements, $\cC$ is a constraint of $E$, $c:E \rightarrow \mathbb{R}_{\geq 0}$ is a cost function, $p:E \rightarrow \mathbb{R}_{\geq 0}$ is a profit function, and $\beta \in \mathbb{R}_{\geq 0}$ is a budget. If $\cC$ is a matching constraint then $I$ is a BM instance; otherwise, $I$ is a BI instance. A {\em solution} of $I$ is a feasible set $S \in \cm(\cC)$ such that $c(S) \leq \beta$. The objective is to find a solution $S$ of $I$ such that $p(S)$ is maximized.
Let $|I|$ denote the encoding size of a BC instance $I$, and $\textnormal{poly}(|I|)$ be a polynomial size in $|I|$.

\section{The Algorithm}
\label{sec:alg}

In this section we present an EPTAS for the BC problem. 
Our first step is to determine the set of {\em profitable} elements in the constructed solution.\footnote{A similar approach is used, e.g., in~\cite{grandoni2010approximation,BBGS11,DKS23}.} 
To this end, we generalize the {\em representative set} notion of \cite{DKS23} to the setting of BC.\footnote{The representative set of \cite{DKS23} can be viewed as
a special case of our {\em strict representative set} (see Definition~\ref{def:Representatives}).} 
Our scheme relies on initially finding a set of profitable elements of small cardinality, from which 
the most profitable elements are selected for the solution using enumeration. 
Then, {\em non-profitable} elements are added to the solution 
using a techniques of~\cite{BBGS11}. 
 
For the remainder of this section, fix a BC instance $I = (E, \cC, c,p, \beta)$ and an error parameter $0<\eps <\frac{1}{2}$. 
Let $H(I,\eps) = \{e \in E~|~ p(e) > \eps \cdot \OPT(I)\}$ be the set of {\em profitable} elements in $I$, and $E \setminus H(I,\eps)$ the set of {\em non-profitable} elements; when understood from the context, we use $H = H(I,\eps)$. Now, a representative set is a subset of elements which contains the profitable elements of an {\em almost} optimal solution. Formally,
\begin{definition}
	\label{def:REP}
		Let $I = (E, \cC, c,p, \beta)$ be a \textnormal{BC} instance, $0<\eps<\frac{1}{2}$ and $R \subseteq E$. 
		We say that $R$ is a {\em representative set} of $I$ and $\eps$ if there is a solution $S$ of $I$ such that the following holds.
	\begin{enumerate}
		\item $S \cap H \subseteq R$. 
		\item $p\left(S\right) \geq (1-4\eps) \cdot \OPT(I)$.
	\end{enumerate} 
	\end{definition}

The work of \cite{DKS23} laid the foundations for the following notions of {\em replacements} and {\em strict representative sets (SRS)}, for the special case of BC where $\cC$ is a single matroid constraint. Below we generalize the
definitions of replacements and SRS. 

Intuitively, a replacement of a solution $S$ for $I$ of bounded cardinality is another solution for $I$ which preserves the attributes of the profitable elements in $S$ (i.e., $S \cap H$). In particular, the profit of the replacement is close to $p(S \cap H)$, whereas the cost and the number of profitable elements can only be smaller. An SRS is a subset of elements containing a replacement for any solution for $I$ of bounded cardinality. 

The formal definitions of replacement and SRS for general BC instances are given in Definitions~\ref{def:Replacement} and~\ref{def:Representatives}, respectively. Let $q(\eps) = \ceil{\eps^{-\eps^{-1}}}$, and $\cm_{\leq q(\eps)} = \{A \in  \cm ~|~|A| \leq q(\eps)\}$ be all {\em bounded feasible sets} of $\cC$ and $\eps$. Recall that we use $\cm = \cm(\cC)$ for the feasible sets of $\cC$; similar simplification in notation is used also for bounded feasible sets.
 \begin{definition}
	\label{def:Replacement}
	Given a \textnormal{BC} instance $I = (E, \cC, c,p, \beta), 0<\eps<\frac{1}{2}$, $S \in \cm_{\leq q(\eps)}$, and $Z_S \subseteq E$, we say that $Z_S$ is a {\em replacement} of $S$ for $I$ and $\eps$ if the following holds: 
	\begin{enumerate}
		
		\item $(S \setminus H) \cup Z_S \in \cm_{\leq q(\eps)}$.\label{p:I}
		
		\item $c(Z_S) \leq c(S \cap H)$.\label{p:s}
		
		\item $p\left((S \setminus H) \cup Z_S\right) \geq (1-\eps) \cdot p(S)$.\label{p:p}
		
		\item $|Z_S| \leq |S \cap H|$.\label{p:car}
	\end{enumerate}
\end{definition}

 \begin{definition}
\label{def:Representatives}
	Given a \textnormal{BC} instance $I = (E, \cC, c,p, \beta), 0<\eps<\frac{1}{2}$, and $R \subseteq E$, we say that $R$ is a {\em strict representative set (SRS)} of $I$ and $\eps$ if, for any $S \in \cm_{\leq q(\eps)}$, there is a replacement $Z_S \subseteq R$ of $S$ for $I$ and $\eps$. 
\end{definition}

Observe that given any solution $S$ of $I$ such that $|S| \leq q(\eps)$, it holds that $S \cap H$  is a replacement of $S$; also, $E$ is an SRS. In the next result, we demonstrate the power of SRS in solving BC. Specifically, we show that
any SRS $R \subseteq E$ is also a representative set. Hence, using enumeration on subsets of $R$ we can find a subset of elements that can be extended by only non-profitable elements to an {\em almost} optimal solution (see Algorithm~\ref{alg:EPTAS}). 
 
\begin{lemma}
\label{lem:Solution}
Let  $I = (E, \cC, c,p, \beta)$  be a \textnormal{BC} instance, let  $0<\eps<\frac{1}{2}$,  and 
let $R$ be an \textnormal{SRS} of $I$ and $\eps$. Then $R$ is a representative set of $I$ and $\eps$. 
\end{lemma}

The proof of Lemma~\ref{lem:Solution} is given in Appendix~\ref{sec:omitted-alg}.
We proceed to construct an  SRS whose cardinality depends only on $\eps$. First, we partition the profitable elements (and possibly some more elements) into a small number of {\em profit classes}, where elements from the same profit class have {\em similar} profits. 
The profit classes are derived from a $2$-approximation $\alpha$ for $\OPT(I)$, which can be easily computed in polynomial time. 
Specifically, for all $r \in [\log_{1-\eps} \left(\frac{\eps}{2}\right)+1]$ define the $r$-{\em profit class} as 
\begin{equation}
	\label{Er}
{\ck}_{r}(\alpha)
 = \left\{e \in E~\bigg|~ \frac{ p(e)}{2 \cdot \alpha} \in \big( (1-\eps)^{r}, (1-\eps)^{r-1} \big]\right\}.
\end{equation} 
 
In the following, we give a definition of an {\em exchange set} for each profit class.  This facilitates the construction of an SRS.   
In words, a subset of elements $X$ is an exchange set for some profit class ${\ck}_r(\alpha)$ if any feasible set $\Delta$ and element $a \in (\Delta \cap {\ck}_r(\alpha)) \setminus X$ can be replaced (while maintaining feasibility) by some element $b \in (X \cap {\ck}_r(\alpha)) \setminus \Delta$ such that the cost of $b$ is no larger than the cost of $a$. Formally, 
\begin{definition}
 	\label{def:r-set}
 	Let $I = (E, \cC, c,p, \beta)$ be a \textnormal{BC} instance, $0<\eps<\frac{1}{2}$, $\frac{\OPT(I)}{2} \leq \alpha \leq \OPT(I)$, $r \in [\log_{1-\eps}\left(\frac{\eps}{2}\right)+1]$, and $X \subseteq {\ck}_r(\alpha)$. We say that $X$ is an {\em exchange set} for $I,\eps,\alpha,$ and $r$ if: 
 \begin{itemize}
\item For all $\Delta \in \cm_{\leq q(\eps)}$ and $a \in (\Delta \cap {\ck}_r(\alpha)) \setminus X$ there is $b \in ({\ck}_r(\alpha) \cap X) \setminus \Delta$ satisfying 
\begin{itemize}
\item $c(b) \leq c(a)$.
\item $\Delta-a+b \in \cm_{\leq q(\eps)}$.
\end{itemize} 
\end{itemize}

 \end{definition}
 
 The similarity between SRS and exchange sets is not coincidental. We show that if a set $R \subseteq E$ satisfies that $R \cap {\ck}_r(\alpha)$ is an exchange set for any $r \in [\log_{1-\eps}\left(\frac{\eps}{2}\right)+1]$, then $R$ is an SRS, and thus also a representative set by Lemma~\ref{lem:Solution}. This allows us to construct an SRS using a union of disjoint exchange sets, one for each profit class. 
 \begin{lemma}
\label{lem:sufficientRep}
	Let $I = (E, \cC, c,p, \beta)$ be a \textnormal{BC} instance, $0<\eps<\frac{1}{2}$, $\frac{\OPT(I)}{2} \leq \alpha \leq \OPT(I)$ and $R \subseteq E$. If for all $r \in [\log_{1-\eps}\left(\frac{\eps}{2}\right)+1]$ it holds that $R \cap {\ck}_r(\alpha)$ is an  exchange set for $I,\eps,\alpha,$ and $r$, then $R$ is a representative set of $I$ and $\eps$. 
\end{lemma} 

We give the formal proof in Appendix~\ref{sec:omitted-alg}.  
We now present a unified algorithm for finding a representative set for both types of constraints, namely, matching or matroid intersection constraints. This is achieved by taking the union of exchange sets of all profit classes. Nevertheless, for the construction of exchange sets we distinguish between the two types of constraints.
This results also in different sizes for 
the obtained representative sets. Our algorithms for finding the exchange sets are the core technical contribution of this paper. 

For matching constraints, we design an algorithm which constructs an exchange set for any profit class by finding multiple matchings of $\cC$ from the given profit class. Each matching has a bounded cardinality, and the edges are chosen using a greedy approach to minimize the cost. 
We give the full details and a formal proof of Lemma~\ref{lem:mainMatching} in Section~\ref{sec:lemMainProofMatching}. 
\begin{lemma}
	\label{lem:mainMatching}
	There is an algorithm $\textnormal{\textsf{ExSet-Matching}}$ that given a \textnormal{BM} instance $I$, $0<\eps <\frac{1}{2}$, $\frac{\OPT(I)}{2} \leq \alpha \leq \OPT(I)$, and $r \in [\log_{1-\eps}\left(\frac{\eps}{2}\right)+1]$, 
	returns in time $q(\eps) \cdot \textnormal{poly}(|I|)$ an exchange set $X$ for $I,\eps,\alpha,$ and $r$, such that $|X| \leq 18 \cdot {q(\eps)}^2$. 
\end{lemma} 

Our algorithm for matroid intersection constraints is more involved and generates an exchange set by an {\em asymmetric interpretation} of the two given matroids. 
 We give the full details in Section~\ref{sec:lemMainProof} and a formal proof of Lemma~\ref{lem:mainMatroid} in Appendix~\ref{sec:omitted-matroid}.
\begin{lemma}
	\label{lem:mainMatroid}
	There is an algorithm $\textnormal{\textsf{ExSet-MatroidIntersection}}$ that given a \textnormal{BI} instance $I$, $0<\eps <\frac{1}{2}$, $\frac{\OPT(I)}{2} \leq \alpha \leq \OPT(I)$, and $r \in [\log_{1-\eps}\left(\frac{\eps}{2}\right)+1]$, 
	returns in time  ${q(\eps)}^{O(q(\eps))} \cdot \textnormal{poly}(|I|)$ an exchange set $X$ for $I,\eps,\alpha,$ and $r$, such that $|X| \leq {q(\eps)}^{O(q(\eps))}$. 
\end{lemma}

 Using the above, we design 
 an algorithm that returns a representative set for both types of constraints. This is done by computing the $2$-approximation $\alpha$ of $\OPT(I)$, and then finding exchange sets for  all profit classes, for the corresponding type of constraint. Finally, we return the union of the above exchange sets. The pseudocode of our algorithm, 
 {\sf RepSet}, is given in Algorithm~\ref{alg:findRep}. 
 \begin{algorithm}[h]
 	\caption{$\textsf{RepSet}(I = (E, \cC, c,p, \beta),\eps)$}
 	\label{alg:findRep}
\SetKwInOut{Input}{input}
\SetKwInOut{Output}{output}

\Input{A BC instance $I$ and error parameter $0<\eps<\frac{1}{2}$.}

\Output{A representative set $R$ of $I$ and $\eps$.}

Compute a $2$-approximation $S^*$ for $I$ using a PTAS for BC with parameter $\eps' = \frac{1}{2}$.\label{step:REPAPP}

Set $\alpha \leftarrow p(S^*)$. \label{step:REPalpha}

	Initialize $R \leftarrow \emptyset$.

	\For{$r \in [\log_{1-\eps} \left(\frac{\eps}{2}\right)+1]$\label{step:REPforr}}{

 \eIf{$I$\textnormal{ is a BM instance}}{$R \leftarrow R \cup \textsf{ExSet-Matching}(I,\eps,\alpha,r)$\label{step:match}.  
 
}{$R \leftarrow R \cup \textsf{ExSet-MatroidIntersection}(I,\eps,\alpha,r)$\label{step:matroid}.

}
}

Return $R$.\label{step:REPret}
 \end{algorithm}

 \begin{lemma}
 	\label{lem:main}
 	  Given a \textnormal{BC} instance $I = (E, \cC, c,p, \beta)$ and $0<\eps <\frac{1}{2}$, Algorithm ~\ref{alg:findRep}  
 	returns a representative set $R$ of $I$ and $\eps$, such that one of the following holds. \begin{itemize}
 		\item If $\cC$ is a matching constraint the running time is $q(\eps)^2 \cdot \textnormal{poly}(|I|)$, and $|R| \leq 54 \cdot q(\eps)^3$. 
 		
 		\item If $\cC$ is a matroid intersection constraint the running time is ${q(\eps)}^{O(q(\eps))} \cdot \textnormal{poly}(|I|)$, and $|R| \leq {q(\eps)}^{O(q(\eps))}$. 
 	\end{itemize} 
 \end{lemma} The proof of the lemma is given in Appendix~\ref{sec:omitted-alg}. 
 Next, we use a result of \cite{BBGS11} for adding elements of smaller profits to the solution. The techniques of \cite{BBGS11} are based on a non-trivial patching of two solutions of the Lagrangian relaxation of BC (both for matching and matroid intersection constraints). This approach yields a feasible set with almost optimal profit, where in the worst case the difference from the optimum is twice the 
 maximal profit of an element in the instance. Since we use the latter approach only for non-profitable elements, this effectively does not harm our approximation guarantee. The following is a compact statement of the above result 
 of~ \cite{BBGS11}.
\begin{lemma}
	\label{lem:grandoni}
	 There is a polynomial-time algorithm $\textnormal{\textsf{NonProfitableSolver}}$ that given a \textnormal{BC} instance $I = (E, \cC, c,p, \beta)$ computes a solution $S$ for $I$ of profit $p(S) \geq \OPT(I)-2 \cdot \max_{e \in E} p(e)$.
\end{lemma}

Using the algorithm above and our algorithm for computing a representative set, we obtain an EPTAS for BC. Let $R$ be the representative set returned by $\textsf{RepSet} (I,\eps)$. Our scheme enumerates over subsets of $R$ to select profitable elements for the solution. Using algorithm $\textnormal{\textsf{NonProfitableSolver}}$ of \cite{BBGS11}, the solution is extended to include also non-profitable elements. Specifically, let $\frac{\OPT(I)}{2} \leq \alpha \leq \OPT(I)$ be a $2$-approximation for the optimal profit for $I$. In addition, let $E(\alpha) = \{e \in E~|~ p(e) \leq 2\eps \cdot \alpha\}$ be the set 
including the non-profitable elements, and possibly also profitable elements $e \in E$ 
such that $p(e) \leq 2 \eps \cdot \OPT(I)$. Given a feasible set $F \in \cm$, we define a residual BC instance containing elements which can {\em extend} $F$ by adding elements from $E(\alpha)$. More formally, 

\begin{definition}
	\label{def:instance}
	Given a \textnormal{BC} instance $I = (E, \cC, c,p, \beta)$, $\frac{\OPT(I)}{2} \leq \alpha \leq \OPT(I)$, and $F \in \cm(\cC)$, the {\em residual instance of $F$ and $\alpha$} for $I$ is the \textnormal{BC} instance $I_F(\alpha) = (E_F,\cC_F,c_F,p_F,\beta_F)$ defined as follows. 
	\begin{itemize}
		
		\item $E_F = E(\alpha) \setminus F$.

		\item $\cC_F = \cC / F$. 
		
		\item $p_F = p|_F$ (i.e., the restriction of $p$ to $F$). 
		
			\item $c_F = c|_F$.
		
		\item $\beta_F = \beta-c(F)$. 
	\end{itemize}

\end{definition}
\begin{observation}
	\label{ob:residual}
	Let $I = (E, \cC, c,p, \beta)$ be a \textnormal{BC} instance, $\frac{\OPT(I)}{2} \leq \alpha \leq \OPT(I)$, $F \in \cm(\cC)$, and let $T$ be a solution for $I_F(\alpha)$. Then, $T \cup F$ is a solution for $I$. 
\end{observation}

For all solutions $F \subseteq R$ for $I$ with $|F| \leq \eps^{-1}$, we find a solution $T_F$ for the residual instance $I_F(\alpha)$ using Algorithm $\textnormal{\textsf{NonProfitableSolver}}$ and define $K_F =  T_F \cup F$ as the  {\em extended solution} of $F$. 
Our scheme iterates over the extended solutions $K_F$ for all such solutions $F$ and chooses
an extended solution $K_{F^*}$ of maximal total profit. The pseudocode of the scheme is given in Algorithm~\ref{alg:EPTAS}.	

\begin{algorithm}[h]
	\caption{$\textsf{EPTAS}(I = (E, \cC, c,p, \beta),\eps)$}
	\label{alg:EPTAS}
	
	\SetKwInOut{Input}{input}
	
	\SetKwInOut{Output}{output}
	
	\Input{A BC instance $I$ and an error parameter $0<\eps<\frac{1}{2}$.}
	
	\Output{A solution for $I$.}
	
	Construct the representative set $R \leftarrow \textsf{RepSet} (I,\eps)$.\label{step:rep}
	
	Compute a $2$-approximation $S^*$ for $I$ using a PTAS for BC with parameter $\eps' = \frac{1}{2}$.\label{step:APP2}
	
	Set $\alpha \leftarrow p(S^*)$.

	Initialize an empty solution $A \leftarrow \emptyset$.\label{step:init}
	
	\For{$F \subseteq R \textnormal{ s.t. } |F| \leq \eps^{-1} \textnormal{ and } F \textnormal{ is a solution of } I $ \label{step:for}}{

		Find a solution for $I_F(\alpha)$ by $T_F \leftarrow \textnormal{\textsf{NonProfitableSolver}}(I_F(\alpha))$.\label{step:vertex}
		
		Let $K_F \leftarrow T_F \cup F$.\label{step:Cf}

		\If{$p\left(K_F\right) > p(A)$\label{step:iff}}{
			
			Update $A \leftarrow K_F$\label{step:update}
			
		}

	}
	
	Return $A$.\label{step:retA}
\end{algorithm}

The running time of Algorithm~\ref{alg:EPTAS} crucially depends on the cardinality of the representative set. Roughly speaking, the running time 
is the number of subsets of the representative set containing at most $\eps^{-1}$ elements, multiplied by a computation time that is polynomial
in the encoding size of the instance. Moreover, since $R = \textsf{RepSet} (I,\eps)$ is a representative set (by Lemma~\ref{lem:main}), there is an almost optimal solution $S$ of $I$ such that the profitable elements in $S$ are a subset of $R$. Thus, there is an iteration of the {\bf for} loop in Algorithm~\ref{alg:EPTAS} such that $F = S \cap H$. In the proof of Lemma~\ref{thm:EPTAS} we focus on this  iteration and show that it yields a solution $K_{F}$ of $I$ with an almost optimal profit. 

\begin{lemma}
	\label{thm:EPTAS}
	Given a \textnormal{BC} instance $I = (E, \cC, c,p, \beta)$ and $0<\eps<\frac{1}{2}$, Algorithm~\ref{alg:EPTAS} returns a solution for $I$ of profit at least $(1-8\eps) \cdot \OPT(I)$ such that one of the following holds. \begin{itemize}
		
			\item  If $I$ is a \textnormal{BM} instance the running time is $2^{ O \left(\eps^{-2} \log \frac{1}{\eps} \right)} \cdot \textnormal{poly}(|I|)$.
		
		\item If $I$ is a \textnormal{BI} instance the running time is ${q(\eps)}^{O(\eps^{-1} \cdot q(\eps))} \cdot \textnormal{poly}(|I|)$, where $q(\eps) = \ceil{\eps^{-\eps^{-1}}}$. 
		
	\end{itemize}
\end{lemma}

The proof of Lemma~\ref{thm:EPTAS} is given in Appendix~\ref{sec:omitted-alg}. We can now prove our main results.

\noindent{\bf Proofs of Theorem~\ref{thm:matching} and Theorem~\ref{thm:matroid}:}  Given a BC instance $I$ and $0<\eps<\frac{1}{2}$, using Algorithm~\ref{alg:EPTAS} for $I$ with parameter $\frac{\eps}{8}$ we have by Lemma~\ref{thm:EPTAS}
the desired approximation guarantee. Furthermore,  
the running time is $2^{ O \left(\eps^{-2} \log \frac{1}{\eps} \right)} \cdot \textnormal{poly}(|I|)$ or ${q(\eps)}^{O(\eps^{-1} \cdot q(\eps))} \cdot \textnormal{poly}(|I|)$, depending on whether $I$ is a BM instance or a BI instance, respectively. 
\qed

\section{Exchange Set for Matching Constraints}
\label{sec:lemMainProofMatching}

In this section we design an algorithm for finding an exchange set for a BM instance and a profit class, leading to the proof of Lemma~\ref{lem:mainMatching}. For the remainder of this section, fix a BM instance	$I = (E, \cC, c,p, \beta)$, an error parameter $0<\eps<\frac{1}{2}$, a $2$-approximation for $\OPT(I)$, $\frac{\OPT(I)}{2} \leq \alpha \leq \OPT(I)$, and an index $r \in [\log_{1-\eps}\left(\frac{\eps}{2}\right)+1]$ of the profit class ${\ck}_r(\alpha)$. 

We note that for a single matroid constraint an exchange set can be constructed by finding a minimum cost basis in the matroid~\cite{DKS23}. More specifically, given a matroid $\cG = (E,\cI)$, it is shown in \cite{DKS23} that 
a minimum cost basis in the matroid $[\cG \cap {\ck}_r(\alpha)]_{\leq q(\eps)}$ is an exchange set for ${\ck}_r(\alpha)$. Such exchange set can be easily computed using a greedy approach. An analogue for the setting of matching constraints is to find a matching of cardinality $\Omega(q(\eps))$ 
and minimum total cost in ${\ck}_r(\alpha)$.
However, as shown in Figure~\ref{pic:example}, this idea fails. Thus, 
we turn to use a completely different approach.

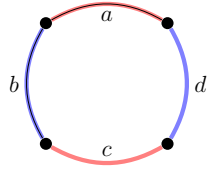
\begin{figure} 	\label{fig:1}
	
	\begin{center}
	\scalebox{0.8}{
		\begin{tikzpicture}
			\node[dot] (0) {};
			\node[dot] [right of=0] (1) {};
			\node[dot] [below of=0] (3) {};
			\node[dot] [right of=3] (4) {};

			\path (4) edge[bend right] node [font=\fontsize{11}{0},right] {$d$} (1);
			\path (4) edge[draw,line width=2pt,-,blue!50,bend right] (1);
			
			\path (0) edge[draw,line width=2pt,-,red!50,bend left] (1);
			\path (0) edge[bend left] node [font=\fontsize{11}{0},below, sloped] {$a$} (1);
			
			\path (0) edge[draw,line width=2pt,-,blue!50,bend right] (3);
			\path (0) edge[bend right] node [font=\fontsize{11}{0},left] {$b$} (3);

			\path (4) edge[bend left] node [font=\fontsize{11}{0},above, sloped] {$c$} (3);
			\path (4) edge[draw,line width=2pt,-,red!50,bend left] (3);
		\end{tikzpicture}
	}	
	\end{center}
\caption{An example showing that bipartite matching may not yield an exchange set. Consider the two matchings $\Delta_1 = \{a,c\}, \Delta_2 =  \{b,d\}$ marked in red and blue, and suppose that ${\ck}_r(\alpha) = \{a,b\}$ is a profit class. The only exchange set for ${\ck}_r(\alpha)$ is $\{a,b\}$, which is not a matching. Note that a bipartite matching can be cast as matroid intersection. For a bipartite graph $G=(L \cup R, E)$, define the matroids ${\cm}_1=(E, {\cI}_1)$ and ${\cm}_2=(E, {\cI}_2)$, where
${\cI}_1= \{ F \subseteq E |~ \forall v \in L: |F \cap N(v)| \leq 1 \}$, and
${\cI}_2= \{ F \subseteq E |~ \forall v \in R: |F \cap N(v)| \leq 1 \}$,
where $N(v)$ is the set of neighbors of $v$.
Thus, bipartite matching is a special case of both matching and matroid intersection.
\label{pic:example}}
\end{figure}
\negA
A key observation is that even if a greedy matching algorithm may not suffice for the construction of an exchange set, applying such an algorithm multiple times can be the solution. Thus, as a subroutine our algorithm finds a matching using a greedy approach. The algorithm iteratively selects an edge of minimal cost
while ensuring that the selected set of edges is a matching.
This is done until the algorithm reaches a given cardinality bound, or no more edges can be added. 
The pseudocode of $\textsf{GreedyMatching}$ is given in Algorithm~\ref{alg:greedyMatching}.\footnote{Given a graph $G = (V,E)$ and a matching $M$ of $G$, the definition of thinning $E / M$ is given in Section~\ref{sec:preliminaries}.}

\begin{algorithm}[h]
	\caption{$\textsf{GreedyMatching}(G = (V,E), N, c)$}
	\label{alg:greedyMatching}

\SetKwInOut{Input}{input}
\SetKwInOut{Output}{output}

\Input{A graph $G$, an integer $N \in \mathbb{N} \setminus \{0\}$, and a cost function $c: E \rightarrow \mathbb{R}_{\geq 0}$.}

\Output{A matching $M$ of $G$.}

Initialize $M \leftarrow \emptyset$.\label{step:greedy:init}

\While{$|M| < N$ \textnormal{and} $E / M \neq \emptyset$\label{step:greedy:while}}{

Find $e \in E / M$ of minimal cost w.r.t. $c$.\label{step:greedy:e}

Update $M \leftarrow M+e$.\label{step:greedy:update}

}

Return $M$.\label{step:greedy:return}
\end{algorithm}

Given a graph $G = (V,E)$ and two edges $a,b \in E$, we say that $a,b$ are {\em adjacent} if there are $x,y,z \in V$ such that $a = \{x,y\}$ and $b = \{y,z\}$; for all $e \in E$, let $\textsf{Adj}_G(e)$ be the set of edges adjacent to $e$ in $G$.  In the next result we show that if an edge $a$ is not selected for the solution by
\textsf{GreedyMatching}, then either the algorithm selects an adjacent edge of cost at most $c(a)$, or all of the selected edges have costs at most $c(a)$.
\begin{lemma}
	\label{lem:GreedyMatching}
	Given a graph $G = (V,E)$, $N \in \mathbb{N} \setminus \{0\}$, and $c:E \rightarrow \mathbb{R}_{\geq 0}$, Algorithm~\ref{alg:greedyMatching} returns in polynomial time a matching $M$ of $G$ such that for all $a \in E \setminus M$ one of the following holds. 
	\begin{enumerate}
		\item $|M| \leq N$ and there is $b \in \textnormal{\textsf{Adj}}_G(a) \cap M$ such that  $c(b) \leq c(a)$.\label{cond:match1}
		\item $|M| = N$, for all $b \in M$ it holds that $c(b) \leq c(a)$, and $M+a$ is a matching of $G$.\label{cond:match2} 
	\end{enumerate} 
\end{lemma}
\begin{proof}
Clearly, Algorithm~\ref{alg:greedyMatching} returns in polynomial time a matching $M$ of $G$.   Observe that $|M| \leq N$ by 
Step~\ref{step:greedy:while}. 
 To prove that either \ref{cond:match1}. or ~\ref{cond:match2}. hold, we distinguish between two cases.
 \begin{itemize}
  		\item   $a \notin E / M$. Then $\textnormal{\textsf{Adj}}_G(a) \cap M \neq \emptyset$. Let $e$ be the first edge in $ \textnormal{\textsf{Adj}}_G(a) \cap M$ that is added to $M$ in Step~\ref{step:greedy:update}; also, let $L$ be the set of edges added to $M$ before $e$. 
  		Then $a \in E / L$, since $L$ does not contain edges adjacent to $a$. By Step~\ref{step:greedy:e}, it holds that $c(e) = \min_{e' \in E / L} c(e') \leq c(a)$.
  	
  	\item $a \in E / M$.  Thus, $|M| = N$; otherwise, $a$ would be added to $M$. Also, $M+a$ is a matching of $G$. Now, let $b \in M$, and let $K$ be the set of edges added to $M$ before $b$. Since $M+a$ is a matching of $G$, by the hereditary property of $(E,\cm(G))$ it holds that $K+a$ is a matching of $G$; thus, $a \in E / K$ and by Step~\ref{step:greedy:e} it follows that $c(b) = \min_{e' \in E / K} c(e') \leq c(a)$. 
  \end{itemize} 
\end{proof}

By Lemma~\ref{lem:GreedyMatching}, we argue that an exchange set can be found by using Algorithm $\textsf{GreedyMatching}$ iteratively. Specifically, let $k(\eps) = 6 \cdot q(\eps)$ and $N(\eps) = 3 \cdot q(\eps)$. We run Algorithm $\textsf{GreedyMatching}$ for $k(\eps)$ iterations, each iteration with a bound $N(\eps)$ on the cardinality of the matching. In iteration $i$, we choose a matching $M_i$ from the edges of the profit class ${\ck}_r(\alpha)$ and remove the chosen edges from the graph. Therefore, in the following iterations, edges adjacent to previously chosen edges can be chosen as well. A small-scale illustration of the algorithm is presented in Figure~\ref{fig:match}. The pseudocode of Algorithm $\textsf{ExSet-Matching}$, which computes an exchange set for the given profit class, is presented in Algorithm~\ref{alg:representativeMatching}.

\begin{algorithm}[h]
	\caption{$\textsf{ExSet-Matching}(I = (E, \cC, c,p, \beta),\eps,\alpha,r)$}
	\label{alg:representativeMatching}

	\SetKwInOut{Input}{input}
	\SetKwInOut{Output}{output}
	
	\Input{ a matching-\textnormal{BC} instance $I$ , $0<\eps<\frac{1}{2}$, $\frac{\OPT(I)}{2} \leq \alpha \leq \OPT(I)$, $r \in [\log_{1-\eps}\left(\frac{\eps}{2}\right)+1]$.}
	
	\Output{An exchange set for $I,\eps,\alpha$, and $r$.}

	Initialize $X \leftarrow \emptyset$ and $\mathcal{E}_0 \leftarrow {\ck}_r(\alpha)$.\label{step:match:init}
	
	\For{$i \in \{1,\ldots, k(\eps)\}$\label{step:match:for}}{
		
		Define $G_i = (V,\mathcal{E}_{i-1})$ where $V$ is the vertex set of $\cC$.\label{step:mathc:Gi}
		
		Compute $M_i \leftarrow \textsf{GreedyMatching}\left(    G_i ,  N(\eps), c|_{\mathcal{E}_{i-1}}  \right)$.\label{step:match:M}
		
		Update $X \leftarrow X \cup M_i$ and  define $\mathcal{E}_i \leftarrow \mathcal{E}_{i-1} \setminus M_i$.\label{step:match:Update}

	}
	
	Return $X$.\label{step:match:return}

\end{algorithm}
\begin{figure}
		\begin{center}
	\label{pic:matching}
	\scalebox{0.9}{
		\begin{tikzpicture}
			\node[dot] (0) {};
			\node[dot] [right of=0] (1) {};
			\node[dot] [below right of=1, yshift=15] (2) {};
			\node[dot] [below of=0] (3) {};
			\node[dot] [right of=3] (4) {};
			
			\node[dot] [right of=1, , xshift=15] (5) {};
			\node[dot] [right of=4, , xshift=15] (6) {};
			
			\path (0) edge[draw,line width=2pt,-,green!50,bend left] (1);
			\path (0) edge[bend left] node [font=\fontsize{9}{0},above, sloped] {$1$} (1);

			\path (0) edge[draw,line width=2pt,-,blue!50,bend right] (3);
			\path (0) edge[bend right] node [font=\fontsize{9}{0},above, sloped] {$1$} (3);
			
			\path (1) edge node [font=\fontsize{9}{0},above, sloped] {$5$} (2);
			\path (1) edge[draw,line width=2pt,-,blue!50] (4);
			\path (1) edge node [font=\fontsize{9}{0},above, sloped] {$1$} (4);
			
			\path (3) edge [draw,line width=2pt,-,red!50] (1);
			\path (3) edge  node [font=\fontsize{9}{0},above, sloped] {$1$} (1); 
			\path (3) edge[bend right] node [font=\fontsize{9}{0},above, sloped] {$4$} (4);
			\path (4) edge [draw,line width=2pt,-,red!50] (2);

			\path (5) edge[draw,line width=2pt,-,blue!50,bend left] (6);
			\path (5) edge[bend left]  node [font=\fontsize{9}{0},above, sloped] {$2$} (6);

			\path (4) edge[draw,line width=2pt,-,green!50,bend right] (6);
			\path (4) edge[bend right] node [font=\fontsize{9}{0},above, sloped] {$2$} (6);
			
			\path (2) edge[draw,line width=2pt,-,green!50] (5);
			\path (5) edge node [font=\fontsize{9}{0},above, sloped] {$3$} (2);
			
			\path (4) edge node [font=\fontsize{9}{0},above, sloped] {$1$} (2);
			\path (1) edge[bend left] node [font=\fontsize{9}{0},above, sloped] {$3$} (5);
			\path (6) edge node [font=\fontsize{9}{0},above, sloped] {$4$} (2);
		\end{tikzpicture}
	}	
		\end{center}
\caption{An illustration of Algorithm \textsf{ExSet-Matching} with the (illegally small) parameters $N(\eps) = k(\eps) = 3$. The parameters by the edges are the costs.  The edges chosen in iterations $i = 1,2,3$ are marked in blue, red, and green, respectively.\label{fig:match}}
\end{figure}
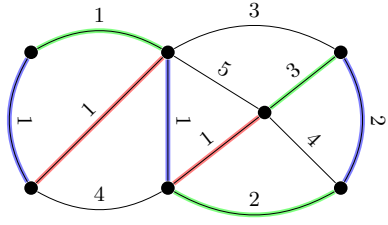

Algorithm \textsf{ExSet-Matching} outputs a union $X$ of disjoint matchings $M_1, \ldots, M_{k(\eps)}$ taken from the edges of the profit class ${\ck}_r(\alpha)$. For some $\Delta \in \cm(\cC)$ and $a \in (\Delta \cap {\ck}_r(\alpha)) \setminus X$, by Lemma~\ref{lem:GreedyMatching}, there are two options summarizing the main idea in the proof of Lemma~\ref{lem:mainMatching}. 
\begin{itemize}
	\item all matchings $M_i$ contain some $b_i$ adjacent to $a$ such that $c(b_i) \leq c(a)$. Then, as $k(\eps)$ is sufficiently large, one such $b_i$ is not adjacent to any edge in $\Delta - a$. Hence, $\Delta-a+b_i$ is a matching. 
	
	\item One such $M_i$ contains only edges of costs at most $c(a)$;
	as $N(\eps)$ is sufficiently large, there is $b \in M_i$ such that $\Delta-a+b$ is a matching.
\end{itemize}

\noindent{\bf Proof of Lemma~\ref{lem:mainMatching}:} For all $i \in \{1,\ldots, k(\eps)\}$, let $G_i$ and $M_i$ be the outputs of Steps~\ref{step:mathc:Gi} and~\ref{step:match:M} in iteration $i$ of the {\bf for} loop in $\textsf{ExSet-Matching}(I,\eps,\alpha,r)$, respectively. Also, let $X$ be the output of the algorithm; observe that $X = \bigcup_{i \in [k(\eps)]} M_i$. 
We show that $X$ is an exchange set for $I,\eps,\alpha$ and $r$ (see Definition~\ref{def:r-set}). Let $\Delta \in \cm_{\leq q(\eps)}$ and $a \in (\Delta \cap {\ck}_r(\alpha)) \setminus X$. We use the next inequality in the claim below. \begin{equation}
	\label{eq:match:1}
	  \frac{k(\eps)}{2} =  N(\eps) = 3 \cdot q(\eps) > 2 \cdot |\Delta| = |V(\Delta)|. 
\end{equation}
The inequality holds since $\Delta \in \cm_{\leq q(\eps)}$. The last equality holds since each vertex appears as an endpoint in a matching at most once.
 \begin{claim}
	\label{claim:adjacent2}
There is $b \in (X \cap {\ck}_r(\alpha)) \setminus \Delta$ such that $\Delta-a+b \in \cm_{\leq q(\eps)}$, and $c(b) \leq c(a)$.
\end{claim}
\begin{claimproof}
Let $a = \{x,y\}$, $I = (E, \cC, c,p, \beta)$, and $\cC = (V,E)$. 
Since $a \notin X$, for all $i \in \{1,\ldots, k(\eps)\}$ it holds that $a \notin M_i$; thus, $a \in \mathcal{E}_{i} = \mathcal{E}_{i-1} \setminus M_i$. Hence, by Lemma~\ref{lem:GreedyMatching}, one of the following holds. 
\begin{enumerate}
	\item For all $i \in [k(\eps)]$ there is $b_i \in \textnormal{\textsf{Adj}}_{G_i}(a) \cap M_i$ such that  $c(b_i) \leq c(a)$. For $z \in \{x,y\}$ let 
	\begin{equation*}
		\label{eq:zz}
		J_z = \left\{i \in [k(\eps)]~\middle|~\exists u \in V:~ b_i = \{z,u\}\right\}
	\end{equation*} be the set of indices of edges $b_i$ neighboring to $z$.  Since $b_i \in \textsf{Adj}_{G_i}(a)$ it holds that $J_x \cup J_y = [k(\eps)]$. Thus, there is $z \in \{x,y\}$ such that $|J_z| \geq \frac{k(\eps)}{2} > |V(\Delta)|$, where the last inequality follows from  \eqref{eq:match:1}. For any $i\in J_z$ let $v_i\in V$ be the vertex connected to $z$ in $b_i$, that is $b_i =\{z, v_i\}$. Since the matchings $M_1,\ldots M_{k(\eps)}$ are disjoint and $b_i \in M_i$ it follows that the vertices $v_i$ for $i\in J_z$ are all distinct. As $|J_z|>|V(\Delta)|$ there is $i^*\in J_z$ such that $v_{i^*}\notin V(\Delta)$. Therefore, $\Delta -a + b_{i^*} \in \cm_{\leq q(\eps)}$ and $c(b_{i^*})\leq c(a) $.
	\item \label{case:2mi} There is  $i \in \{1,\ldots, k(\eps)\}$ such that $|M_i| = N(\eps)$, for all $b \in M_i$ it holds that $c(b) \leq c(a)$, and $M_i+a$ is a matching of $G_i$. Then, 
	\begin{equation}
		\label{eq:Nvv}
		|M_i| = N(\eps) > |V(\Delta)|.
	\end{equation} 
The equality follows by the definition of $M_i$ in Case~\ref{case:2mi}. The inequality follows from~\eqref{eq:match:1}. Since each vertex appears as an endpoint in a matching at most once, by \eqref{eq:Nvv} there is $b \in M_i$ such that both endpoints of $b$ are not in $V(\Delta)$. Thus, $\Delta+b \in \cm$; by the hereditary property and since $a \in \Delta$, it holds that $\Delta-a+b \in \cm_{\leq q(\eps)}$. 
\end{enumerate}
\end{claimproof} 
By Claim~\ref{claim:adjacent2} and Definition~\ref{def:r-set}, we have that $X$ is an exchange set for $I,\eps,\alpha$, and $r$ as required. 
To complete the proof of the lemma we show (in Appendix~\ref{sec:omitted-match})
the following.
	\begin{claim}
	\label{claim:RunningMatch}
	$|X| \leq 18 \cdot q(\eps)^2$, and the running time of  Algorithm~\ref{alg:representativeMatching} is $q(\eps)^{O(q(\eps))} \cdot \textnormal{poly}(|I|)$.
\end{claim} \qed

\section{Exchange Set for Matroid Intersection Constraints}
\label{sec:lemMainProof}

In this section, we design an algorithm for finding an exchange set 
for a profit class in a BI instance,
leading to the proof of Lemma~\ref{lem:mainMatroid}. For the remainder of this section, fix a BI instance $I = (E, \cC, c,p, \beta)$, an error parameter $0<\eps<\frac{1}{2}$, a $2$-approximation for $\OPT(I)$, $\frac{\OPT(I)}{2} \leq \alpha \leq \OPT(I)$, and an index $r \in [\log_{1-\eps}\left(\frac{\eps}{2}\right)+1]$ of the profit class ${\ck}_r(\alpha)$. Also, let $\cC = (\cI_{1},\cI_{2})$ be the matroid intersection constraint $\cC$ of $I$. For simplicity, when understood from the context, some of the lemmas in this section consider the given parameters (e.g., $I$) without an explicit declaration. Due to space constraints, the proofs of the lemmas in this section are given in Appendix~\ref{sec:omitted-matroid}. 

As shown in Figure~\ref{pic:example}, a simple greedy approach which finds a feasible set of minimum cost (within ${\ck}_r(\alpha)$) in the intersection of the matroids may not output an exchange set for ${\ck}_r(\alpha)$. 
Instead, our approach builds on some interesting properties of matroid intersection.
The next definition presents a {\em shifting property} for a feasible set $\Delta \in \cm_{\leq q(\eps)}$ and an element $a \in \Delta \cap {\ck}_r(\alpha)$ 
 w.r.t. the two matroids. We use this property to show that our algorithm constructs an exchange set. 
 \begin{definition}
	\label{def:shift}
	Let $\Delta \in \cm_{\leq q(\eps)}$, $a \in \Delta \cap {\ck}_r(\alpha)$ and $b \in {\ck}_r(\alpha) \setminus \Delta$. We say that $b$ is a {\em shift} to $a$ for $\Delta$ if $c(b) \leq c(a)$ and $\Delta-a+b \in \cm_{\leq q(\eps)}$; moreover, $b$ is a {\em semi-shift} to $a$ for $\Delta$ if $c(b) \leq c(a)$ and $\Delta-a+b \in \cI_{2}$ but $\Delta-a+b \notin \cI_{1}$. 
\end{definition}

As a starting point for our exchange set algorithm, we show how to obtain small cardinality sets which contain either a shift or a semi-shift for every pair $\Delta \in \cm_{\leq q(\eps)}$ and $a \in \Delta \cap {\ck}_r(\alpha)$. The proof of the following is given in Appendix~\ref{sec:omitted-matroid}.

\begin{lemma}
\label{lem:Thereb}
Let $U \subseteq {\ck}_r(\alpha)$, $\Delta \in \cm_{\leq q(\eps)}$, and $B$ be a minimum basis of $[(E,\cI_2) \cap U]_{\leq q(\eps)}$ w.r.t. $c$. Also, let  $a \in  (U \cap \Delta) \setminus B$. Then, there is $b \in B \setminus \Delta$ such that $b$ is a semi-shift to $a$ for $\Delta$ or $b$ is a shift to $a$ for $\Delta$. 
\end{lemma} 
	
Observe that to have an exchange set, our goal is to find a subset of ${\ck}_r(\alpha)$ which contains a shift for every pair $\Delta \in \cm_{\leq q(\eps)}$ and $a \in \Delta \cap {\ck}_r(\alpha)$. Thus, using Lemma~\ref{lem:Thereb} we design the following recursive algorithm $\textsf{ExtendChain}$
, which finds a union of minimum bases of matroids w.r.t $\cI_2$, of increasingly restricted ground sets w.r.t. $\cI_1$. The pseudocode of Algorithm \textsf{ExtendChain} is given in Algorithm~\ref{alg:representative}.

We can view the execution of $\textsf{ExtendChain}$ as a tree, where each node (called below a {\em branch}) corresponds to the subset $S \subseteq {\ck}_r(\alpha)$ in specific recursive call. We now describe the role of $S$ in Algorithm $\textsf{ExtendChain}$.
If $|S| \geq q(\eps)+1$, we simply return $\emptyset$; such a  branch is called a  {\em leaf}, and does not contribute elements to the constructed exchange set. Otherwise, define the {\em universe} of the branch $S$ as $U_S =  \{e \in {\ck}_r(\alpha) \setminus S~|~S+e \in \cI_1\}$; that is, elements in the universe of $S$ can be added to $S$ to form an independent set w.r.t. $\cI_1$.
Next, we construct a minimum basis $B_{S}$ w.r.t. $c$ of the matroid $[(E,\cI_2) \cap U_{S}]_{\leq q(\eps)}$.
Observe that $B_S$ contains up to $q(\eps)$ elements, taken from the universe of $S$ and that $B_S$ is independent w.r.t. $\cI_2$. Note that the definition of the universe relates to $\cI_1$ while the construction of the bases to $\cI_2$; thus, the two matroids play completely different roles in the algorithm. 

For every element $e \in B_S$ we apply Algorithm \textsf{ExtendChain} recursively with $S' = S+e$ to find the corresponding basis $B_{S+e}$. The algorithm returns (using recursion) the union of the constructed bases over all branches. Finally, algorithm \textsf{ExSet-MatroidIntersection} constructs an exchange set for $I,\eps,\alpha$, and $r$ by computing Algorithm $\textsf{ExtendChain}$ with the initial branch (i.e., {\em root}) $S = \emptyset$: 
\begin{equation}
	\label{eq:Ex-set}
	\textsf{ExSet-MatroidIntersection}(I,\eps,\alpha,r) = 	\textsf{ExtendChain}(I,\eps,\alpha,r,\emptyset). 
\end{equation} For an illustration of the algorithm, see Figure~\ref{pic:example2}. 

\begin{figure} 	\label{fig:12}
	\begin{center}
	\scalebox{0.9}{
		\begin{tikzpicture}
			\node[circle,draw,double] (empt) {$\emptyset$};
			\node[circle,draw,double] [below left of=empt,yshift=15,xshift=-30] (a) {$\{a\}$};
			\node[circle,draw,double] [below right of=empt,yshift=15,xshift=30] (b) {$\{b\}$};
			\node[circle,draw,double,scale=0.8] [below left of=b] (be) {$\{b,e\}$};
			\node[circle,draw,double,scale=0.8] [below right of=b] (bf) {$\{b,f\}$};
			\node[circle,draw,double,scale=0.8] [below left of=a] (ac) {$\{a,c\}$};
			\node[circle,draw,double,scale=0.8] [below right of=a] (ad) {$\{a,d\}$};
		
			\draw[edge] (empt) -- (a);
			\draw[edge] (empt) -- (b);
			\draw[edge] (b) -- (be);
			\draw[edge] (b) -- (bf);
			\draw[edge] (a) -- (ad);
			\draw[edge] (a) -- (ac);
		
		\end{tikzpicture}
	}
	\end{center}	
	\caption{An illustration of the branches in Algorithm~\ref{alg:representative} for $S = \emptyset$.
		Note that $B_{\emptyset} = \{a,b\}$, $B_{\{a\}} = \{c,d\}$ and $B_{\{b\}} = \{e,f\}$. Also, $\{a,c\}$ and $\{a,d\}$ are the child branches of $\{a\}$. 	
		\label{pic:example2}} 
\end{figure}
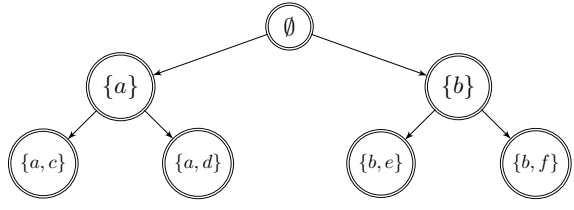

\begin{algorithm}[h]
	\caption{$\textsf{ExtendChain}(I = (E, \cC, c,p, \beta),\eps,\alpha,r,S)$}
	\label{alg:representative}
		\SetKwInOut{Input}{input}
	\SetKwInOut{Output}{output}
	
		\Input{ a matroid-\textnormal{BC} instance $I$ , where $\cC = (\cI_1,\cI_2)$, $0<\eps<\frac{1}{2}$, $\frac{\OPT(I)}{2} \leq \alpha \leq \OPT(I)$, $r \in [\log_{1-\eps}\left(\frac{\eps}{2}\right)+1]$, and $S \subseteq E$.}
	
	\Output{(for $S = \emptyset$) An exchange set $X$ for $I,\eps,\alpha$, and $r$.}

		\If{$|S| \geq q(\eps)+1$\label{step:q}}{
		
Return $\emptyset$
	
}
		
			Define $U_S =  \{e \in {\ck}_r(\alpha) \setminus S~|~S+e \in \cI_1\}$.\label{step:U_S}
		
			Compute a minimum basis $B_S$ w.r.t. $c$ of the matroid $\left[ (E, \cI_2) \cap U_S \right]_{\leq q(\eps)}$.\label{step:Br}
		
		Return $ B_S \cup \left(\bigcup_{e \in B_S} \textsf{ExtendChain}(I,\eps,\alpha,r,S+e)\right)$.\label{step:RBS}

\end{algorithm}

In the analysis of the algorithm, we consider branches with useful attributes, called {\em chains}; these are essentially sequences of semi-shifts to some $\Delta \in \cm_{\leq q(\eps)}$ and $a \in \Delta \cap {\ck}_r(\alpha)$. 
 Let $X = 	\textsf{ExSet-MatroidIntersection}(I,\eps,\alpha,r)$, and let $\mathcal{S}$ be the set of all branches $S \subseteq {\ck}_r(\alpha)$ such that $\textsf{ExtendChain}(I,\eps,\alpha,r,S)$ is computed during the construction of $X$. 
 
 \begin{definition}
	\label{def:chain}
	Let $S \in \mathcal{S}$, $\Delta \in \cm_{\leq q(\eps)}$, and $a \in ({\ck}_r(\alpha) \cap \Delta) \setminus X$. We say that $S$ is a {\em chain} of $a$ and $\Delta$ if $a \in U_S$, and for all $e \in S$ it holds that $e$ is a semi-shift to $a$ for $\Delta$.
	\end{definition}

Note that there must be a chain for $a$ and $\Delta$ since the empty set satisfies the conditions of Definition~\ref{def:chain}. Moreover, 
we can bound the cardinality of a chain by $q(\eps)$ using the exchange property of the matroid $(E,\cI_{1})$. The above arguments are formalized in the next lemmas. 

	\begin{lemma}
	\label{claim:emptyChain}
	For all $\Delta \in \cm_{\leq q(\eps)}$ and $a \in ({\ck}_r(\alpha) \cap \Delta) \setminus X$ there is $S \subseteq X$ such that $S$ is a chain of $a$ and $\Delta$. 
\end{lemma}

	\begin{lemma}
	\label{claim:qChain}
	For all $\Delta \in \cm_{\leq q(\eps)}$, $a \in ({\ck}_r(\alpha) \cap \Delta) \setminus X$, and a chain $S$ of $a$ and $\Delta$, it holds that $|S| \leq q(\eps)$. 
\end{lemma}

For a chain $S$ of $a$ and $\Delta$, let $B_S$ be the result of the first computation of Step~\ref{step:Br} (i.e., not in a recursive call) in $\textsf{ExtendChain}(I,\eps,\alpha,r,S)$. The key argument in the proof of Lemma~\ref{lem:mainMatroid} is that for a chain $S^*$ of maximal cardinality, $B_{S^*}$ contains a shift to $a$ and $\Delta$, using the maximality of $S^*$ and Lemma~\ref{lem:Thereb}. 
	\begin{lemma}
	\label{lem:maximum}
		For all $\Delta \in \cm_{\leq q(\eps)}$, $a \in ({\ck}_r(\alpha) \cap \Delta) \setminus X$, and a chain $S^*$ of $a$ and $\Delta$ of maximum cardinality, there is a shift $b^* \in B_{S^*}$ to $a$ for $\Delta$.  
\end{lemma}

In the proof of Lemma~\ref{lem:mainMatroid}, for every $\Delta \in \cm_{\leq q(\eps)}$ and $a \in ({\ck}_r(\alpha) \cap \Delta) \setminus X$ we take a chain $S^*$ of $a$ and $\Delta$ of maximum cardinality (which exists by Lemma~\ref{claim:emptyChain} and Lemma~\ref{claim:qChain}). Then, by Lemma~\ref{lem:maximum}, there is a shift $b^*$ to $a$ for $\Delta$, and it follows 
that $X$ is an exchange set for $I,\eps,\alpha$, and $r$. The formal proof is given in Appendix~\ref{sec:omitted-matroid}.

\section{Discussion}
\label{sec:discussion}

In this paper we present the first EPTAS for budgeted matching and 
budgeted matroid intersection, thus improving upon the existing PTAS for both problems. We derive our results via a generalization of the representative set framework of Doron-Arad et al. \cite{DKS23}; this ameliorates the exhaustive enumeration applied in similar settings \cite{BBGS11,chekuri2011multi}. 

We note that the framework based on representative sets may be useful for 
solving other problems formulated as~\eqref{eq:1}.
Indeed, the proofs of Lemma~\ref{lem:Solution} and Lemma~\ref{lem:sufficientRep}, which establish the representative set framework, are oblivious to the exact type of constraints and only require having a $k$-{\em exchange system} for some constant $k$.\footnote{A set system $(E,\cI)$ satisfies the $k$-exchange property if for all $A \in \cI$ and $e \in E$ there is $B \subseteq A, |B| \leq k$, such that $(A \setminus B) \cup \{e\} \in \cI$.}

Furthermore, our exchange sets algorithms can be applied with slight modifications to other variants of our problems and are thus of independent interest. 
In particular, we can use a generalization of Algorithm~\ref{alg:representativeMatching} to construct an exchange set for the {\em budgeted b-matching} problem. Also, we believe that Algorithm~\ref{alg:representative} can be generalized to construct exchange sets for budgeted {\em multi-matroid intersection} for any constant number of matroids; this includes the {\em budgeted multi-dimensional matching} problem. While this problem does not admit a PTAS  unless $\textnormal{P=NP}$~\cite{kann1991maximum}, our initial study shows that by constructing a representative set we may obtain an {\em \textnormal{FPT}-approximation scheme} by parameterizing on the number of elements in the solution.\footnote{We refer the reader, e.g., to~\cite{M08} for the definition of parameterized approximation algorithms running in fixed-parameter tractable (FPT)-time.} 

Finally, to resolve the complexity status of BM and BI, the gripping question of whether the problems admit an FPTAS needs to be answered. Unfortunately, this may be a very difficult task. Even for special cases of a single matroid, such as graphic matroid, the existence of an FPTAS is still open.
Moreover, a deterministic FPTAS for budgeted matching would solve deterministically the exact matching problem, which has been open for over four decades~\cite{papadimitriou1982complexity}.

\bibliography{budgeted}

\begin{thebibliography}{10}

\bibitem{BBGS11}
Andr{\'e} Berger, Vincenzo Bonifaci, Fabrizio Grandoni, and Guido Sch{\"a}fer.
\newblock Budgeted matching and budgeted matroid intersection via the gasoline
  puzzle.
\newblock {\em Mathematical Programming}, 128(1):355--372, 2011.

\bibitem{camerini1992random}
Paolo~M. Camerini, Giulia Galbiati, and Francesco Maffioli.
\newblock Random pseudo-polynomial algorithms for exact matroid problems.
\newblock {\em Journal of Algorithms}, 13(2):258--273, 1992.

\bibitem{chekuri2011multi}
Chandra Chekuri, Jan Vondr{\'a}k, and Rico Zenklusen.
\newblock Multi-budgeted matchings and matroid intersection via dependent
  rounding.
\newblock In {\em Proceedings of the twenty-second annual ACM-SIAM symposium on
  Discrete Algorithms}, pages 1080--1097. SIAM, 2011.

\bibitem{cormen2022introduction}
Thomas~H Cormen, Charles~E Leiserson, Ronald~L Rivest, and Clifford Stein.
\newblock {\em Introduction to algorithms}.
\newblock MIT press, 2022.

\bibitem{DKS23}
Ilan Doron-Arad, Ariel Kulik, and Hadas Shachnai.
\newblock An {EPTAS} for budgeted matroid independent set.
\newblock In {\em Symposium on Simplicity in Algorithms (SOSA)}, pages 69--83,
  2023.

\bibitem{garey1979computers}
Michael~R Garey and David~S Johnson.
\newblock {\em Computers and intractability}, volume 174.
\newblock freeman San Francisco, 1979.

\bibitem{AAM_2ndEd}
Teofilo~F. Gonzalez, editor.
\newblock {\em Handbook of Approximation Algorithms and Metaheuristics, Second
  Edition, Volume 1: Methologies and Traditional Applications}.
\newblock Chapman and Hall/CRC, 2018.

\bibitem{grandoni2010approximation}
Fabrizio Grandoni and Rico Zenklusen.
\newblock Approximation schemes for multi-budgeted independence systems.
\newblock In {\em European Symposium on Algorithms}, pages 536--548. Springer,
  2010.

\bibitem{kann1991maximum}
Viggo Kann.
\newblock Maximum bounded 3-dimensional matching is {MAX} {SNP}-complete.
\newblock {\em Information Processing Letters}, 37(1):27--35, 1991.

\bibitem{kulik2010there}
Ariel Kulik and Hadas Shachnai.
\newblock There is no {EPTAS} for two-dimensional knapsack.
\newblock {\em Information Processing Letters}, 110(16):707--710, 2010.

\bibitem{M08}
D{\'a}niel Marx.
\newblock Parameterized complexity and approximation algorithms.
\newblock {\em The Computer Journal}, 51(1):60--78, 2008.

\bibitem{mulmuley1987matching}
Ketan Mulmuley, Umesh~V Vazirani, and Vijay~V Vazirani.
\newblock Matching is as easy as matrix inversion.
\newblock In {\em Proceedings of the nineteenth annual ACM symposium on Theory
  of computing}, pages 345--354, 1987.

\bibitem{papadimitriou1982complexity}
Christos~H Papadimitriou and Mihalis Yannakakis.
\newblock The complexity of restricted spanning tree problems.
\newblock {\em Journal of the ACM (JACM)}, 29(2):285--309, 1982.

\bibitem{pferschy2009knapsack}
Ulrich Pferschy and Joachim Schauer.
\newblock The knapsack problem with conflict graphs.
\newblock {\em J. Graph Algorithms Appl.}, 13(2):233--249, 2009.

\bibitem{ravi1996constrained}
Ram Ravi and Michel~X Goemans.
\newblock The constrained minimum spanning tree problem.
\newblock In {\em Scandinavian Workshop on Algorithm Theory}, pages 66--75.
  Springer, 1996.

\bibitem{schrijver2003combinatorial}
Alexander Schrijver et~al.
\newblock {\em Combinatorial optimization: polyhedra and efficiency},
  volume~24.
\newblock Springer, 2003.

\bibitem{SW01}
Petra Schuurman and Gerhard~J Woeginger.
\newblock Approximation schemes-a tutorial.
\newblock {\em Lectures on scheduling}, 2001.

\end{thebibliography}
\appendix

\section{Omitted Proofs from Section~\ref{sec:alg}}
\label{sec:omitted-alg}

We give several definitions and claims used in the proof of Lemma~\ref{lem:Solution}. The next general observation follows from applying the exchange property of matroids multiple times. It is used in the proof of Lemma~\ref{ob:matroid}. 
\begin{observation}
	\label{obr:matroid}
	Given a matroid $(E,\cI)$ and $A,B \in \cI$, there is $D \subseteq A \setminus B$, $|D| = \max\{|A|- |B|,0\}$ such that $B \cup D \in  \cI$. 
\end{observation} 
A {\em set system} is a tuple $(E,\cI)$, where $E$ is a ground set of elements and $\cI \subseteq 2^E$ is a set of subsets of $E$ (e.g., a matroid is a set system). A set system $(E,\cI)$ satisfies the {\em hereditary property} if for all $A \in \cI$ and $B \subseteq A$ it holds that $B \in \cI$. Note that the set system $(E,\cm(\cC))$, where $\cC$ is a constraint of $E$, satisfies the hereditary property both for matching constraints and for matroid intersection constraints. 
The next result describes a {\em weak exchange property} for both types of constraints. The proof relies on the fact that both matchings and matroid intersection have similar (yet weaker) exchange properties to matroids.
\begin{lemma}
	\label{ob:matroid}
	Given a a ground set $E$, a constraint $\cC$ of $E$, and $A,B \in \cm$, there is $D \subseteq A \setminus B$, $|D| = \max\{|A|-2 \cdot |B|,0\}$ such that $B \cup D \in \cm$. 
\end{lemma}
\begin{proof}
	
	Consider the following two cases. \begin{enumerate}
		\item $\cC$ is a matroid intersection constraint. Let $\cC = (\cI_1, \cI_2)$. 
		By Observation~\ref{obr:matroid}, there are $D_1,D_2 \subseteq A \setminus B$ such that $|D_1| = |D_2| =  \max\{|A| - |B|,0\}$, $B \cup D_1 \in  \cI_1$, and $B \cup D_2 \in  \cI_2$. Observe that $D_1 \cap D_2 \subseteq D_1$ and $D_1 \cap D_2 \subseteq D_2$; hence, by the hereditary properties of $(E,\cI_1)$ and $(E,\cI_2)$ it holds that $B \cup (D_1 \cap D_2) \in  \cI_1 \cap \cI_2 = \cm$. Moreover, \begin{equation*}
			\begin{aligned}
				|D_1 \cap D_2| ={} & |A \setminus \left( (A \setminus D_1) \cup (A \setminus D_2)\right)| \\ \geq{} & |A|-|A \setminus D_1|-|A \setminus D_2| \\ 
				={} & |A|-(|A| - |D_1|)-(|A| - |D_2|) \\ 
				\geq{} & |A|-(|A|-(|A|-|B|))-(|A|-(|A|-|B|)) \\
				={} & |A|- 2 \cdot |B|.
			\end{aligned}
		\end{equation*}
		The first and second equalities hold since $D_1, D_2 \subseteq A$. The second inequality follows from $|D_1| = |D_2| =  \max\{|A| - |B|,0\}$. Hence, we have that $D = D_1 \cap D_2$ satisfies the conditions of the lemma.
		
		\item $\cC$ is a matroid intersection constraint. 
		Let $\cC = (V,E)$ 
		and define $D = \{\{u,v\} \in A~|~u,v \notin V(B)\}$. Clearly, $B \cup D$ is a matching of $G$ since $A,B$ are matchings of $G$. In addition, \begin{equation*}
			\begin{aligned}
				|D| ={} & |A \setminus (A \setminus D)| \\
				={} & |A|-\left|\left\{\{u,v\} \in A~|~u \in V(B) \text{ or } v \in V(B)\right\}\right| \\ \geq{} & |A| - |V(B)| \\
				={} & |A|-2 \cdot |B|.
			\end{aligned}
		\end{equation*} The inequality holds since a vertex can be an endpoint in $A$ at most once (recall that $A$ is a matching). The last equality follows by the observation that $|V(B)| = 2 \cdot |B|$, since each edge in a matching has exactly two endpoints, and the endpoints of distinct edges in a matching are disjoint.
	\end{enumerate}
\end{proof}

\noindent{\bf Proof of Lemma~\ref{lem:Solution}}: 
As this proof is based on the proof of a corresponding Lemma in \cite{DKS23}, we use some of their definitions and notations. With a slight abuse of notation, 
we use $\OPT$ to denote an optimal solution for $I$.
Given an optimal solution, we partition a subset of the elements in the solution 
into $\ceil{\eps^{-1}}$ disjoint sets (some sets may be empty). Specifically,
let $N = \ceil{\eps^{-1}}$; for all $i \in [N]$ define 
\begin{equation}
	\label{J_i}
	J_i = \big\{e \in \OPT~\big|~ p(e) \in \big(\eps^{i} \cdot \OPT(I),  \eps^{i-1} \cdot \OPT(I) \big] \big\}.		
\end{equation} 
Let $i^* = \argmin_{i\in [N]} p(J_i)$. By \eqref{J_i} we have at least 
$\eps^{-1}$ disjoint sets; thus, $p(J_{i^*}) \leq \eps \cdot \OPT(I)$. 
Now, let $L~= ~\bigcup_{k \in  [i^*-1]} J_k$ be the subset of all elements in $\OPT$ 
of profits greater than $\eps^{i^*-1} \cdot \OPT(I)$, and $Q = \OPT \setminus (L \cup J_{i^*})$. 
To complete the proof of the lemma, we need several claims.
\begin{claim}
	\label{clam:Iq}
	$L \in \cm_{\leq q(\eps)}$. 
\end{claim}
\begin{claimproof}
	Since $L \subseteq \OPT$, by the hereditary property of $\cm$ we have that $L \in  \cm$. Also, 
	\begin{equation*}
		\label{contradiction1}
		|L| \leq \sum_{e \in L} \frac{p(e)}{\eps^{i^*-1} \cdot \OPT(I)} = \frac{p(L)}{\eps^{i^*-1} \cdot \OPT(I)} \leq  \eps^{- (i^{*}-1)}  \leq   \eps^{-N+1} \leq \eps^{-\eps^{-1}} \leq  q(\eps)
	\end{equation*} The first inequality holds since $p(e) \geq  \eps^{i^*-1} \cdot \OPT$
	for all $e \in L$. For the second inequality, we note that $L$ is a solution for $I$. By the above it follows that $L \in \cm_{\leq q(\eps)}$. 
\end{claimproof} 
By Claim~\ref{clam:Iq} and as $R$ is an SRS,
it follows that $L$ has a replacement $Z_L \subseteq R$. 
Let $\Delta_L = (L \setminus H) \cup Z_L$. By Property~\ref{p:I} of Definition~\ref{def:Replacement}, we have that $\Delta_L \in \cm_{\leq q(\eps)}$, and it also holds that $\cm_{\leq q(\eps)} \subseteq  \cm$. 
Hence, $\Delta_L \in  \cm$. Furthermore, as $Q \subseteq \OPT$ and  $\OPT \in \cm$, by the 
hereditary property for $\cm$ we have that $Q \in  \cm$. Therefore, by Lemma~\ref{ob:matroid}, there is a subset $T \subseteq Q \setminus\Delta_L$, where $|T|  = \max\{|Q| - 2 \cdot |\Delta_L|,0\}$, such that $\Delta_L \cup T \in \cm$. 

Let $S = \Delta_L \cup T$. We show that $S$ satisfies the conditions of the lemma. 

\begin{claim}
	\label{claim:IsSolution}
	$S$ is a solution for $I$. 
\end{claim}

\begin{claimproof}
	By the definition of $T$ it holds that $S = \Delta_L \cup T \in \cm$. Moreover, 	
	\[
	\begin{array}{ll}
		c(S) &= c(\Delta_L\cup T) \\
		& \leq c(Z_L)+ c(L \setminus H) +w(T) \\
		& \leq c(L \cap H)+c(L \setminus H)+c(T) \\ 
		& \leq c(L)+c(Q) \\ 
		& \leq c(\OPT)\\
		& \leq \beta.
	\end{array}
	\]
	The second inequality holds since $c(Z_L) \leq c(L \cap H)$ (see Property~\ref{p:s} of Definition~\ref{def:Replacement}). For the third inequality, recall that $T \subseteq Q$. The last inequality holds since $\OPT$ is a solution for $I$. 
\end{claimproof} 

The proof of the next claim relies on
the {\em profit gap}  between the elements in $Q$ and $L$. 
\begin{claim}
	\label{clam:profitBound1}
	$p\left(Q \setminus T\right) \leq 2\eps \cdot \OPT(I)$.
\end{claim}
\begin{claimproof} 
	
	Observe that 
	\begin{equation}
		\label{eq:menS}
		|Q \setminus T| \leq 2 \cdot |\Delta_L| \leq 2 \cdot |Z_L|+2 \cdot |L \setminus H| \leq 2 \cdot |L \cap H|+2 \cdot |L \setminus H| \leq 2 \cdot |L|.
	\end{equation} 
	The first inequality follows from the definition of $T$. For the third inequality, we use Property~\ref{p:car} of Definition~\ref{def:Replacement}. Hence, 
	\begin{equation*}
		\begin{aligned}
			p(Q \setminus T) \leq{} & |Q \setminus T| \cdot \eps^{i^*} \cdot \OPT(I) \leq  2 \cdot |L| \cdot \eps^{i^*} \cdot \OPT(I) \leq 2\eps \cdot p(L) \leq 2\eps \cdot \OPT(I).
		\end{aligned}
	\end{equation*} 
	
	The first inequality holds since $p(e) \leq \eps^{i^*} \cdot \OPT(I)$ for all $e \in Q$. 
	The second inequality is by \eqref{eq:menS}. The third inequality holds since $p(e) > \eps^{i^*-1} \cdot \OPT(I)$ for all $e \in L$. 
\end{claimproof} 	

\begin{claim}
	\label{clam:profitBound2}
	$p\left(S\right) \geq (1-4\eps) \OPT(I)$.
\end{claim} 

\begin{claimproof} 
	By Property~\ref{p:p} of Definition~\ref{def:Replacement},
	\begin{equation}
		\label{eq:proofProfit1}
		\begin{aligned}
			p(\Delta_L) ={} & p((L \setminus H) \cup Z_L) \geq (1-\eps) \cdot p(L).
		\end{aligned}
	\end{equation} Moreover, 
	\begin{equation}
		\label{eq:proofProfit2}
		\begin{aligned}
			p(T) {} & \geq p(Q) - p(Q \setminus T) \geq p(Q) - 2\eps \cdot \OPT(I) \geq (1-\eps) \cdot p(Q) - 2\eps \cdot \OPT(I). 
		\end{aligned}
	\end{equation} 
	
	The second inequality is by Claim~\ref{clam:profitBound1}. Using \eqref{eq:proofProfit1} and \eqref{eq:proofProfit2}, we have that 
	
	\[	
	\begin{array}{ll}
		p(\Delta_L)+p(T) & \geq (1-\eps) \cdot p(L \cup Q)-2\eps \cdot \OPT(I) 
		\\ & = (1-\eps) \cdot p(\OPT \setminus J_{i^*})-2\eps \cdot \OPT(I) \\
		&  \geq (1-\eps) \cdot (1-\eps) \cdot \OPT(I)-2\eps \cdot \OPT(I)\\
		&  \geq (1-4\eps) \cdot \OPT(I). 
	\end{array}
	\]
	
	The second inequality holds since $p(J_{i^*}) \leq \eps \cdot \OPT(I)$. Observe that 
	$S = \Delta_L \cup T$ and $T \cap \Delta_L = \emptyset$. Therefore,  $$p(S) = p(\Delta_L)+p(T) \geq (1-4\eps) \cdot \OPT(I).$$
	
\end{claimproof} 
Finally, we note that, by \eqref{J_i} and the definition of $Q$, it holds that $Q \cap H = \emptyset$. Consequently, by the definition of $S$, we have that $S \cap H \subseteq Z_L$. 
As $Z_L \subseteq R$ (by Definition~\ref{def:Representatives}), it follows that $S \cap H \subseteq R$. Hence, using Claims~\ref{claim:IsSolution} and~\ref{clam:profitBound2}, it follows that $R$ is a representative set by Definition~\ref{def:REP}.

{\bf Proof of Lemma~\ref{lem:sufficientRep}:}
	For the proof of the lemma, we define a {\em substitution} of some feasible set adapting the definition in \cite{DKS23}. The first two properties of substitution of some $G \in \cm_{\leq q(\eps)}$ are identical to the first two properties in the definition of a replacement of $G$. 
	However, we require that a substitution preserves the number of profitable elements in $G$ from each profit class, and that a substitution must be disjoint to the set of non-profitable elements in $G$. %
	
	\begin{definition}
		\label{def:sub}
		For $G \in \cm_{\leq q(\eps)}$ and $Z_G \subseteq \bigcup_{r \in [\log_{1-\eps} \left(\frac{\eps}{2}\right)+1] } {\ck}_r(\alpha)$, we say that $Z_G$ is a {\em substitution} of $G$ if the following holds. \begin{enumerate}

			\item $(G \setminus H) \cup Z_G \in \cm_{\leq q(\eps)}$.\label{pp:1}
			
			\item  $c(Z_G) \leq c(G \cap H)$.\label{pp:2}
			
			\item For all $r \in [\log_{1-\eps} \left(\frac{\eps}{2}\right)+1]$ it holds that $|{\ck}_r(\alpha) \cap Z_G| = |{\ck}_r(\alpha) \cap G \cap H|$.\label{pp:3}

			\item $(G \setminus H) \cap Z_G = \emptyset$.\label{pp:4}
		\end{enumerate}
	\end{definition}
	\begin{claim}
		\label{claim:substitution}
		For any $G \in \cm_{\leq q(\eps)}$ there is a substitution $Z_G$ of $G$ such that $Z_G \subseteq R$. 
	\end{claim}
	\begin{claimproof}
		Let $G \in \cm_{\leq q(\eps)}$ and let $Z_G$ be a substitution of $G$ such that $|Z_G \cap R|$ is maximal among all substitutions of $G$; formally, let $\mathcal{S}(G)$ be all substitutions of $G$ and let $Z_G \in \{ Z \in \mathcal{S}(G)~|~ |Z \cap R| = \max_{Z' \in \mathcal{S}(G)} |Z' \cap R|\}$. Since $G \cap H$ is in particular a substitution of $G$  it follows that $\mathcal{S}(G)\neq \emptyset$, and thus  $Z_G$ is well defined.  Assume towards a contradiction that there is $a \in Z_G \setminus R$; then, by Definition~\ref{def:sub} there is $r \in [\log_{1-\eps} \left(\frac{\eps}{2}\right)+1]$ such that $a \in {\ck}_r(\alpha)$. Let $\Delta_G = (G \setminus H) \cup Z_G$; by definition~\ref{def:sub} it holds that $\Delta_G \in \cm_{\leq q(\eps)}$. In addition, because $R$ is exchange for $I,\eps,\alpha$, and $r$, by Definition~\ref{def:r-set} there is $b \in ({\ck}_r(\alpha) \cap R) \setminus \Delta_G$ such that $c(b) \leq c(a)$ and $\Delta_G -a+b \in \cm_{\leq q(\eps)}$. Then,  the properties of Definition~\ref{def:sub} are satisfied for $Z_G-a+b$ by the following. 
		\begin{enumerate}
			\item $(G\setminus H) \cup (Z_G -a +b) =  \Delta_G-a+b \in \cm_{\leq q(\eps)}$ by the definition of $b$. 
			
			\item  $c(Z_G-a+b) \leq c(Z_G) \leq c(G\cap H)$ because $c(b) \leq c(a)$. 
			
			\item for all $r' \in [\log_{1-\eps} \left(\frac{\eps}{2}\right)+1]$ it holds that $|{\ck}_{r'}(\alpha) \cap (Z_G-a+b)| = |{\ck}_{r'}(\alpha) \cap Z_G| = |{\ck}_{r'}(\alpha) \cap G \cap H|$ because $a,b \in {\ck}_r(\alpha)$.
			
			\item $|(G \setminus H) \cap (Z_G -a+b)| \leq |(G \setminus H) \cap (Z_G)| = 0$ where the inequality follows because $b \notin \Delta_G$ and the equality is since $Z_G$ is a substitution of $G$. 
		\end{enumerate}
		
		By the above and Definition~\ref{def:sub}, it follows that $Z_G+a-b$ is a substitution of $G$; that is, $Z_G+a-b \in \mathcal{S}(G)$. Moreover, \begin{equation}
			\label{eq:ZG}
			|R \cap (Z_G -a+b)|>|R \cap Z_G| = \max_{Z \in \mathcal{S}(G)} |Z \cap R|.
		\end{equation} The first inequality is because $a \in Z_G \setminus R$ and $b \in R$. By \eqref{eq:ZG} we reach a contradiction since we found a substitution of $G$ with more elements in $R$ than $Z_G \in \mathcal{S}(G)$, which is defined as a substitution of $G$ with maximal number of elements in $R$. Therefore, $Z_G \subseteq R$ as required.
		\end{claimproof}

\begin{claim}
	\label{claim:RR}
	For all $e \in H$ such that $\{e\} \in \cm$, there is exactly one $r \in [\log_{1-\eps} \left(\frac{\eps}{2}\right)+1]$ such that $e \in {\ck}_r(\alpha)$.
\end{claim}
\begin{claimproof}
	Let $e \in H$. Observe that: \begin{equation}
		\label{eq:app}
		\frac{\eps}{2} \leq \frac{p(e)}{2 \cdot \OPT(I)} \leq \frac{p(e)}{2 \cdot \alpha} \leq \frac{p(e)}{\OPT(I)} \leq 1.
	\end{equation}The first inequality holds since $e \in H$. The second and third inequalities follow since
	$\frac{\OPT(I)}{2} \leq \alpha \leq \OPT(I)$. The last inequality is because $\{e\}$ is a solution of $I$ (or this element can be discarded from the instance as it cannot appear in any solution of $I$). In addition, observe that for $r_0 = \ceil{\log_{1-\eps} \frac{\eps}{2}}$ it holds that $(1-\eps)^{r_0} \leq \frac{\eps}{2}$ and that for $r_1 = 1$ it holds that $(1-\eps)^{r_1-1} = 1$. Hence, because $r_0 \leq \floor{\log_{1-\eps} \left(\frac{\eps}{2}\right)+1}$, by \eqref{eq:app}, there is exactly one $r \in [\log_{1-\eps} \left(\frac{\eps}{2}\right)+1]$ such that $\frac{p(e)}{2\alpha} \in \big((1-\eps)^r, (1-\eps)^{r-1}\big]$; thus, by \eqref{Er} it holds that $e \in {\ck}_r(\alpha)$ and $e \notin {\ck}_{r'}(\alpha)$ for $r' \in [\log_{1-\eps} \left(\frac{\eps}{2}\right)+1] \setminus \{r\}$.  
\end{claimproof}

 The proof of Lemma~\ref{lem:sufficientRep} follows by showing that for any $G \in \cm_{\leq q(\eps)}$ a substitution of $G$ which is a subset of $R$ is in fact a replacement of $G$. Let $G \in \cm_{\leq q(\eps)}$; by Claim~\ref{claim:substitution}, $G$ has a substitution $Z_G \subseteq R$. Then, \begin{equation}
	\label{eq:profitR}
	\begin{aligned}
		p(Z_G) \geq{} & \sum_{r \in [\log_{1-\eps} \left(\frac{\eps}{2}\right)+1]} p({\ck}_r(\alpha) \cap Z_G) 
		\\ \geq{} &  \sum_{r \in [\log_{1-\eps} \left(\frac{\eps}{2}\right)+1] \text{ s.t. } {\ck}_r(\alpha) \neq \emptyset} |{\ck}_r(\alpha) \cap Z_G| \cdot \min_{e \in {\ck}_r(\alpha)} p(e) 
		\\ \geq{} & \sum_{r \in [\log_{1-\eps} \left(\frac{\eps}{2}\right)+1] \text{ s.t. } {\ck}_r(\alpha) \neq \emptyset} |{\ck}_r(\alpha) \cap G \cap H| \cdot (1-\eps) \cdot \max_{e \in {\ck}_r(\alpha)} p(e) 
		\\ \geq{} & (1-\eps) \cdot p(G \cap H).
	\end{aligned}
\end{equation}  The third inequality is by \eqref{Er} and by Property~\ref{pp:3} of Definition~\ref{def:sub}. The last inequality follows from Claim~\ref{claim:RR} (note that $\{e\} \in \cm$ for all $e \in G \cap H$ by the hereditary property). Therefore,  
\begin{equation}
	\label{eq:profitFINAL}
	\begin{aligned}
		p((G \setminus H) \cup Z_G) & = p(G \setminus H)+p(Z_G) \\
		&  \geq p(G \setminus H)+(1-\eps) \cdot p(G \cap H) \\
		& \geq (1-\eps) \cdot p(G).
	\end{aligned}
\end{equation}
The first equality follows from Property~\ref{pp:4} of Definition~\ref{def:sub}. The first inequality is by \eqref{eq:profitR}. Now,  $Z_G$ satisfies Property~\ref{p:I} and~\ref{p:s} of Definition~\ref{def:Replacement} by Properties~\ref{pp:1},~\ref{pp:2} of Definition~\ref{def:sub}, respectively. In addition, $Z_G$ satisfies Property~\ref{p:p} of Definition~\ref{def:Replacement} by \eqref{eq:profitFINAL}. Finally, $Z_G$ satisfies Property~\ref{p:car}  of Definition~\ref{def:Replacement} by 
$$|Z_G| = \sum_{r \in [\log_{1-\eps} \left(\frac{\eps}{2}\right)+1]} |Z_G \cap {\ck}_r(\alpha)| \leq \sum_{r \in [\log_{1-\eps} \left(\frac{\eps}{2}\right)+1]} |G \cap H \cap {\ck}_r(\alpha)| = |G \cap H|.$$ 
The first inequality holds since $Z_G$ is a substitution of $G$. The last equality follows from Claim~\ref{claim:RR}. We conclude that $Z_G$ is a replacement of $G$ such that $Z_G \subseteq R$; thus, $R$ is a strict representative set by Definition~\ref{def:Representatives}. Therefore, it follows that $R$ is a representative set by Lemma~\ref{lem:Solution}. \qed

	In the following proofs we rely on the next technical auxiliary claim. \begin{claim}
	\label{claim:log}
	For all $0<\eps<\frac{1}{2}$ it holds that $\log_{1-\eps} \left(\frac{\eps}{2}\right)+1 \leq 3 \eps^{-2}$.
\end{claim}
\begin{claimproof}
	Using logarithm rules we have
	\begin{equation*}
		\label{eq:log1}
		\log_{1-\eps} \left(\frac{\eps}{2}\right)+1  \leq 
		\frac{\ln \left(\frac{2}{\eps}\right)}{-\ln \left(1-\eps \right)}+1 \leq \frac{2\eps^{-1}}{\eps}+1 \leq 3 \eps^{-2}. 
	\end{equation*} 
	The second inequality follows from  $x< -\ln (1-x), \forall x>-1, x \neq 0$, and $\ln (y) < y, \forall y>0$. 
\end{claimproof}

\noindent
{\bf Proof of Lemma~\ref{lem:main}:}
	Let $\alpha$ be the result of Step~\ref{step:REPalpha} of Algorithm~\ref{alg:findRep} and $r \in  [\log_{1-\eps} \left(\frac{\eps}{2}\right)+1]$. We distinguish between two cases, depending on the type of constraint. 
	\begin{itemize}
		\item 	If $I$ is a BM instance then, by Lemma~\ref{lem:mainMatching}, $ \textsf{ExSet-Matching}(I,\eps,\alpha,r)$ is an exchange set for $I,\eps,\alpha$, and $r$ such that $\left|     \textsf{ExSet-Matching}(I,\eps,\alpha,r)    \right| \leq 18 \cdot {q(\eps)}^2$. Therefore 
	  $R \cap {\ck}_r(\alpha)$ is an exchange set for $I,\eps,\alpha$, for all $r \in  [\log_{1-\eps} \left(\frac{\eps}{2}\right)+1]$ and \begin{equation*}
	  	|R| \leq \left( \log_{1-\eps} \left(\frac{\eps}{2}\right)+1\right)  \cdot 18 \cdot q(\eps)^2 \leq 3 \cdot \eps^{-2} \cdot 18 \cdot q(\eps)^2 \leq 54 \cdot {q(\eps)}^3. 
	  \end{equation*}  The second inequality follows from Claim~\ref{claim:log}. Then, by Lemma~\ref{lem:sufficientRep}, $R$ is a representative set of $I$ and~$\eps$; furthermore, $|R| \leq 54 \cdot {q(\eps)}^3$. 
		
		As for the running time, observe that Step~\ref{step:REPAPP} can be computed in time $\textnormal{poly}(|I|)$ using a PTAS for BC with
		an error parameter $\eps= \frac{1}{2}$ (see~\cite{BBGS11}). In addition, each iteration of the \textbf{for} loop of Step~\ref{step:REPforr} can be computed in time $q(\eps) \cdot \textnormal{poly}(|I|)$ by Lemma~\ref{lem:mainMatching}. Hence, as we have $\left(\log_{1-\eps} \left(\frac{\eps}{2}\right)\right)+1$ iterations of the  \textbf{for} loop of Step~\ref{step:REPforr} , the running time of the algorithm is bounded by $$ \left( \log_{1-\eps} \left(\frac{\eps}{2}\right)+1\right)  \cdot q(\eps) \cdot  \textnormal{poly}(|I|) \leq 6 \cdot \eps^{-2} \cdot q(\eps) \cdot  \textnormal{poly}(|I|) \leq q(\eps)^2 \cdot \textnormal{poly}(|I|).$$ The first inequality follows from Claim~\ref{claim:log}. 
		\item  If $I$ is a BI instance, then by Lemma~\ref{lem:mainMatroid}, $\textsf{ExSet-MatroidIntersection}(I,\eps,\alpha,r)$ is an exchange set for $I,\eps,\alpha$, and $r$ such that $\left|     \textsf{ExSet-MatroidIntersection}(I,\eps,\alpha,r)    \right| \leq {q(\eps)}^{O(q(\eps))}$. Therefore, by Step~\ref{step:REPforr}, Step~\ref{step:matroid}, and Step~\ref{step:REPret} of Algorithm~\ref{alg:findRep} it holds that $R$ is an exchange set for $I,\eps,\alpha$, for all $r \in  [\log_{1-\eps} \left(\frac{\eps}{2}\right)+1]$ such that $|R| \leq \left( \log_{1-\eps} \left(\frac{\eps}{2}\right)+1\right)  \cdot  {q(\eps)}^{O(q(\eps))}$.  Then, by Definition~\ref{def:r-set} and Lemma~\ref{lem:sufficientRep} it holds that $R$ is a representative set of $I$ and $\eps$. 
		
		Observe that Step~\ref{step:REPAPP} can be computed in time $\textnormal{poly}(|I|)$ using a PTAS for BC with
		an error parameter $\eps= \frac{1}{2}$ (see~\cite{BBGS11}). In addition, each iteration of the \textbf{for} loop of Step~\ref{step:REPforr} can be computed in time ${q(\eps)}^{O(q(\eps))} \cdot \textnormal{poly}(|I|)$ by Lemma~\ref{lem:mainMatroid}. Hence, as we have $O\left(\log_{1-\eps} \left(\frac{\eps}{2}\right)\right) $ iterations of the  \textbf{for} loop of Step~\ref{step:REPforr}, the running time of the algorithm is $O\left(\log_{1-\eps} \left(\frac{\eps}{2}\right)\right)  \cdot {q(\eps)}^{O(q(\eps))} \cdot \textnormal{poly}(|I|)$; then, by Claim~\ref{claim:log} it holds that $O\left(\log_{1-\eps} \left(\frac{\eps}{2}\right)\right) = O(\eps^{-2})$ and it follows that $O\left(\log_{1-\eps} \left(\frac{\eps}{2}\right)\right)  = {q(\eps)}^{O(1)}$. Thus, the running time of the algorithm is $ {q(\eps)}^{O(q(\eps))} \cdot \textnormal{poly}(|I|)$. In addition, since $O\left(\log_{1-\eps} \left(\frac{\eps}{2}\right)\right) = O(\eps^{-2})$ by Claim~\ref{claim:log}, it follows that $|R| \leq {q(\eps)}^{O(q(\eps))}$.
	\end{itemize}  \qed

For the proof of Lemma~\ref{thm:EPTAS}, we use the next auxiliary lemmas.  

\begin{lemma}
	\label{thm:aux1}
	Given a \textnormal{BC} instance $I = (E, \cC, c,p, \beta)$ and $0<\eps<\frac{1}{2}$, Algorithm~\ref{alg:EPTAS} returns a solution for $I$ of profit at least $(1-8\eps) \cdot \OPT(I)$.
\end{lemma}

\begin{proof}

	By Lemma~\ref{lem:main} it holds that $R = \textsf{RepSet}(I,\eps)$ is a representative set of $I$ and $\eps$. Therefore, by Definition~\ref{def:REP} there is a solution $S$ for $I$ such that $S \cap H \subseteq R$, and \begin{equation}
		\label{eq:profitS}
		p\left(S\right) \geq (1-4\eps) \cdot \OPT(I).
	\end{equation} As for all $e \in S \cap H$ we have $p(e) > \eps \cdot \OPT(I)$, and $S$ is a solution for $I$, it follows that $|S \cap H| \leq \eps^{-1}$.
	We note that there is an iteration of Step~\ref{step:for} in which $F = S \cap H$; thus, in Step~\ref{step:vertex} we construct a solution $T_{S \cap H}$ of $I_{S \cap H}$ such that: 
	\begin{equation}
		\label{eq:finalProfitA}
		\begin{aligned}
			p\left(T_F\right) \geq{} & \OPT(I_{S \cap H})-2 \cdot \max_{e \in E_{S \cap H}} p(e) \\
			\geq{} & p(S \setminus H)- 2 \cdot \max_{e \in E_{S \cap H}} p(e) \\
			\geq{} & p(S \setminus H)-4\eps \cdot \OPT(I). 
		\end{aligned}
	\end{equation}
	The first inequality holds by Lemma~\ref{lem:grandoni}. The second inequality holds since $S \setminus H$ is in particular a solution of the residual instance $I_F(\alpha)$ by Definition~\ref{def:instance}. The third inequality holds since for all $e \in E_{S \cap H}$ it holds that $p(e) \leq 2\eps \cdot \alpha \leq 2 \eps \cdot \OPT(I)$. Now, recall that $K_{S \cap H}$ defined in Step~\ref{step:Cf} of Algorithm~\ref{alg:EPTAS}. 
	\begin{equation}
		\label{eq:finalProfit}
		\begin{aligned}
			p(K_{S \cap H}) ={} & p(S \cap H)+p\left(T_F\right)  \geq  p(S)- 4\eps \cdot \OPT(I) \geq (1-8\eps) \cdot \OPT(I).
		\end{aligned}
	\end{equation}
	
	The first inequality uses~\eqref{eq:finalProfitA}. The last inequality is by \eqref{eq:profitS}.  \begin{claim}
		\label{claim:Cf}
		$A  = \textnormal{\textsf{EPTAS}}(\cI,\eps)$ is a solution of $I$. 
	\end{claim}
	\begin{claimproof} If $A = \emptyset$ the claim trivially follows since $\emptyset$ is a solution of $I$. Otherwise, by Step~\ref{step:update} of Algorithm~\ref{alg:EPTAS}, there is a solution $F$ of $I$ such that $A = K_F$. Thus, in this case $A$ is a solution for $I$ by Observation~\ref{ob:residual}.
	\end{claimproof} 
	
	By Claim~\ref{claim:Cf}, Steps~\ref{step:for}, \ref{step:update} and~\ref{step:retA}
	of Algorithm~\ref{alg:EPTAS} and \eqref{eq:finalProfit}, we have that
	$A = \textsf{EPTAS}(I,\eps)$ is a solution for $I$ satisfying $p(A) \geq p(K_{S \cap H}) \geq (1-8\eps) \OPT(I)$. This completes the proof. 	
\end{proof}

\begin{lemma}
	\label{thm:aux2}
	Given a \textnormal{BC} instance $I = (E, \cC, c,p, \beta)$ and $0<\eps<\frac{1}{2}$, the running time of Algorithm~\ref{alg:EPTAS} satisfies one of the following. \begin{itemize}
		
			\item If $I$ is a \textnormal{BM} instance  the running time is $2^{ O \left(\eps^{-2} \log \frac{1}{\eps} \right)} \cdot \textnormal{poly}(|I|)$.
		
		\item If $I$ is a \textnormal{BI} instance the running time is ${q(\eps)}^{O(\eps^{-1} \cdot q(\eps))} \cdot \textnormal{poly}(|I|)$.

	\end{itemize}
\end{lemma}
\begin{proof}

	In the following, let $W' = \big\{F \subseteq R~\big|~ |F| \leq \eps^{-1}, F \in \cm(\cC), c(F) \leq \beta\big\}$ be all feasible sets considered in Step~\ref{step:for} of Algorithm~\ref{alg:EPTAS} and let $W = \big\{F \subseteq R~\big|~ |F| \leq \eps^{-1}\big\}$. Observe that the number of iterations of Step~\ref{step:for} of Algorithm~\ref{alg:EPTAS} is bounded by $|W|$, since $W' \subseteq W$ and for each $F \in W$ we can verify in polynomial time if $F \in W'$. We split the analysis for the upper bound on $|W|$ into two parts. \begin{enumerate}

		\item 	$I$ is a BM instance.

		\begin{equation}
			\label{eq:subR}
			\begin{aligned}
				|W| \leq{} &  \left(|R|+1\right)^{\eps^{-1}}\\
				\leq{} &  \left(54 \cdot {q(\eps)}^3+1\right)^{\eps^{-1}} \\ 
				\leq{} & {\left(\eps^{-6}\right)}^{\eps^{-1}} \cdot \ceil{\eps^{-\eps^{-1}}}^{3 \cdot \eps^{-1}} \\
				\leq{} & {\eps^{-6 \cdot \eps^{-1}-6\eps^{-2}}} \\
				={} & 2^{ O \left(\eps^{-2} \log \frac{1}{\eps} \right)}.
			\end{aligned}
		\end{equation} 
	The second inequality holds by Lemma~\ref{lem:main}. The third inequality holds since $0<\eps<\frac{1}{2}$.

		\item $I$ is a BI instance. 
	Then,
	\begin{equation}
		\label{eq:subR11111}
		\begin{aligned}
			|W| \leq{} &  \left(|R|+1\right)^{\eps^{-1}}
			\leq  {\left({q(\eps)}^{O(q(\eps))}\right)}^{\eps^{-1}} = {q(\eps)}^{O(\eps^{-1} \cdot q(\eps))}.
		\end{aligned}
	\end{equation} 
	The second inequality follows from Lemma~\ref{lem:main}. 
	\end{enumerate} Hence, by \eqref{eq:subR} and \eqref{eq:subR11111}, the number of iterations of the {\bf for} loop  in Step~\ref{step:for} is bounded by $2^{ O \left(\eps^{-2} \log \frac{1}{\eps} \right)}$ and ${q(\eps)}^{O(\eps^{-1} \cdot q(\eps))}$ for BM and BI instances, respectively. In addition, by Lemma~\ref{lem:grandoni}, the running time of each iteration is  $\textnormal{poly}(|I|)$. By the above, the running time of Algorithm~\ref{alg:EPTAS} is $2^{ O \left(\eps^{-2} \log \frac{1}{\eps} \right)} \cdot \textnormal{poly}(|I|)$ if $I$ is a BM instance, and the running time is ${q(\eps)}^{O(\eps^{-1} \cdot q(\eps))} \cdot \textnormal{poly}(|I|)$  if $I$ is a BI instance.  \end{proof}

\noindent{\bf Proof of Lemma~\ref{thm:EPTAS}:} The proof follows from Lemma~\ref{thm:aux1} and Lemma~\ref{thm:aux2}. \qed

\section{Omitted Proofs from Section~\ref{sec:lemMainProofMatching}}
\label{sec:omitted-match}

{\bf Proof of Claim~\ref{claim:RunningMatch}:} For the running time, observe that each iteration of the {\bf for} loop of Step~\ref{step:match:for} can be computed in polynomial time, since Step~\ref{step:mathc:Gi} and Step~\ref{step:match:Update} take linear time and Step~\ref{step:match:M} takes polynomial time by Lemma~\ref{lem:GreedyMatching}. Thus, as we have $k(\eps)  = O(q(\eps))$ iterations, the running time is $q(\eps) \cdot \textnormal{poly}(|I|)$. Finally, for the size of the exchange set $X$: \begin{equation*}
	|X| = \left|  \bigcup_{i \in [k(\eps)]} M_i \right| \leq \sum_{i \in [k(\eps)]} |M_i| \leq k(\eps) \cdot N(\eps) = 18 \cdot q(\eps)^2. 
\end{equation*}The first equality holds by Step~\ref{step:match:Update} and Step~\ref{step:match:return}. The second inequality holds since $\forall i \in [k(\eps)]:~|M_i| \leq N(\eps)$ by Lemma~\ref{lem:GreedyMatching}. The last equality holds since $k(\eps) = 2 \cdot N(\eps) = 6 \cdot q(\eps)$. \qed

\section{Omitted Proofs from Section~\ref{sec:lemMainProof}}
\label{sec:omitted-matroid}

Lemmas~\ref{lem:notIS} and \ref{lem:gen} provide basic matroid properties used in this section; 
the proofs can be found in \cite{DKS23}. 

\begin{lemma}
	\label{lem:notIS}
	Given a matroid $(E,\cI)$ and a cost function $c:E \rightarrow \mathbb{R}_{\geq 0}$, let $B$ be a minimum basis of $(E,\cI)$ w.r.t. $c$. Then, for any $a \in E \setminus B$ it holds that $\{e\in B ~|~ c(e)\leq c(a)\}+ a \notin \cI$. 	 
\end{lemma}

\begin{lemma}
	\label{lem:gen}
	Let $ (E,\cI)$ be a matroid, $A,B \in \cI$, and $a \in A \setminus B$ such that $B+a \notin \cI$. Then there is $b \in B \setminus A$ such that $A-a+b \in \cI$. 
\end{lemma} 

We give a more general formulation of Lemma~\ref{lem:Thereb}. 
\begin{lemma}
	\label{lem:NN}
 Let $\cG = (E,\cI)$ be a matroid, $c: E \rightarrow \mathbb{R}_{\geq 0}$,  $U \subseteq E$, $q \in \mathbb{N}$, and $B$ be a minimum basis of $[ \cG \cap U]_{\leq q}$ w.r.t. $c$. Also, let $\Delta \in \cI_{\leq q }$ and let $a \in (\Delta \cap U) \setminus B$. Then, there is $b \in B \setminus \Delta$ such that $\Delta - a+ b \in \cI$ and $c(b) \leq c(a)$.
\end{lemma}

\begin{proof}

Define $D = \{e \in B~|~c(e) \leq c(a)\}$ and let 
\begin{equation}
	\label{eq:I2q}
	\cI^U_{\leq q} = \{A \subseteq U~|~A \in \cI, |A| \leq q\}
\end{equation} be the collection of independent sets of the matroid $[\cG \cap U]_{\leq q}$ by Definition~\ref{def:matroids}. Note that $a \in  U \setminus B$; in addition, since $D \subseteq B$ and $a \notin B$ it holds that $a \notin D$. Therefore,  it follows that $D+a \notin \cI^U_{\leq q}$ by Lemma~\ref{lem:notIS}. We consider two cases. 
  \begin{itemize}
		\item $D+a \in \cI$. Then, \begin{equation}
			\label{eq:DaI2}
			|D| \geq q \geq |\Delta|>|\Delta-a|.
		\end{equation} 
		The first inequality holds since $D+a  \in \cI$, $D+a \notin \cI^U_{\leq q}$, and $D+a \subseteq U$; thus, $|D+a|> q$ by \eqref{eq:I2q}. The second inequality follows since $\Delta \in \cI_{\leq q }$ and the last inequality holds because $a \in \Delta$. 
		Observe that  by the hereditary property it holds that $D \in \cI^U_{\leq q} \subseteq \cI_{\leq q}$ and that $\Delta-a \in \cI_{\leq q}$ (note that $\Delta \in \cI_{\leq q}$). 
		Thus, by \eqref{eq:DaI2} and the exchange property of  $[\cG]_{\leq q}$ there is $b \in D \setminus (\Delta-a)$ such that $\Delta-a+b \in \cI_{\leq q}$. Because $a \notin D$ it holds that $b \notin \Delta$. It follows that $c(b) \leq c(a)$ using $b \in D$.
		
		\item $D+a \notin \cI$. Observe that $\Delta \in \cI$ by Definition~\ref{def:matroids}; also, observe that $D \in \cI$ because $D \subseteq B$, $ B \in \cI^U_{\leq q}$, and $\cI^U_{\leq q} \subseteq \cI$. Therefore, by Lemma~\ref{lem:gen} there is $b \in D \setminus \Delta$ such that $\Delta-a+b \in \cI_{\leq q}$. Observe that $c(b) \leq c(a)$ because $b \in D$. 
	\end{itemize}
\end{proof} 
\noindent{\bf Proof of Lemma~\ref{lem:Thereb}:} By Lemma~\ref{lem:NN} there is  $b \in B \setminus \Delta$ such that $\Delta-a+b \in \cI_{2}$ and $c(b) \leq c(a)$. Therefore, by Definition~\ref{def:shift} if it also holds that $\Delta-a+b \in \cI_{1}$, then $b$ is a shift to $a$ for $\Delta$; otherwise, $\Delta-a+b \notin \cI_{1}$ and it holds that $b$ is a semi-shift to $a$ for $\Delta$. \qed

~\\

\noindent{\bf Proof of Lemma~\ref{claim:emptyChain}:}
	Let $S = \emptyset$. Trivially, $ \emptyset \subseteq X$ and $S \in \mathcal{S}$ by \eqref{eq:Ex-set}. Moreover, since $a \in \Delta$ and $\Delta \in \cm_{\leq q(\eps)}$, by the hereditary property $a = \emptyset+a \in \cI_1$; thus,  $a \in U_{\emptyset}$ by Step~\ref{step:U_S} of Algorithm~\ref{alg:representative}. Finally, for all $e \in \emptyset$ it holds that $e$ is a semi-shift to $a$ for $\Delta$ as a vacuous truth. Thus, $\emptyset$ is a chain of $a$ and $\Delta$ by Definition~\ref{def:chain}. \qed

~\\

\noindent{\bf Proof of Lemma~\ref{claim:qChain}:}
	Assume towards contradiction that $|S| > q(\eps)$; then, 
	\begin{equation}
		\label{eq:ddd}
		|S| > q(\eps) \geq |\Delta| > |\Delta-a|.
	\end{equation} 
	The second inequality holds since $\Delta \in \cm_{\leq q(\eps)}$. Now, by Definition~\ref{def:chain} $a \in U_{S}$, and by Step~\ref{step:U_S} of Algorithm~\ref{alg:representative} it holds that $S+a \in \cI_1$; therefore, $S \in \cI_1$. In addition, as $\Delta \in \cm_{\leq q(\eps)}$, by the hereditary property $\Delta -a \in \cI_{1}$. Therefore, by \eqref{eq:ddd} and the exchange property of $(E,\cI_{1})$ there is $e \in S \setminus (\Delta-a)$ such that $\Delta-a+e \in \cI_1$. However, as $S$ is a chain and $e \in S$, we have that $e$ is a semi-shift to $a$ for $\Delta$ (see Definitions~\ref{def:shift} and~\ref{def:chain}), implying that $\Delta-a+e \notin \cI_1$. Contradiction. \qed
~\\

\noindent{\bf Proof of Lemma~\ref{lem:maximum}:}
By Lemma~\ref{lem:Thereb}, there is $b^* \in B_{S^*}$ such that $b^*$ is a semi-shift or a shift to $a$ for $\Delta$. Assume towards a contradiction that $b^*$ is a semi-shift to $a$ for $\Delta$. We show that $S^*+b^*$ is a chain of $a$ and $\Delta$. 
\begin{itemize}
	\item $S^*+b^* \in \mathcal{S}$. Follows from Step~\ref{step:RBS} of Algorithm~\ref{alg:representative} and since $|S^*| \leq q(\eps)$, by Lemma~\ref{claim:qChain}.
	\item $a \in U_{S^*+b^*}$. Assume towards contradiction that $a \notin U_{S^*+b^*}$; then, by Lemma~\ref{lem:gen}, there is $e \in S^*+b^*$ such that $\Delta -a+e \in \cI_{1}$. Contradiction (as $S^*$ is a chain of $a$ and $\Delta$, and $b^*$ is a semi-shift to $a$ for $\Delta$).
	\item $\forall e \in S^*+b^*$ it holds that $e$ is a semi-shift to $a$ for $\Delta$. This follows since $S^*$ is a chain of $a$ and $\Delta$ and $b^*$ is a semi-shift to $a$ for $\Delta$, by our assumption. 
\end{itemize}
By the above, $S^*+b^*$ is a chain of $a$ and $\Delta$.
Furthermore, $|S^*+b^*| > |S^*|$ since $b^* \in B_{S^*}	\subseteq U_{S^*}$ and 
$B_{S^*} \cap S^* = \emptyset$. Thus, we have a contradiction to the maximality of $S^*$. We conclude that $b^*$ is a shift to $a$ for $\Delta$. \qed

~\\

\noindent{\bf Proof of Lemma~\ref{lem:mainMatroid}:} Let $\Delta \in \cm_{\leq q(\eps)}$, $a \in (C_r(\alpha) \cap \Delta) \setminus X$, 
and let $S^*$ be a chain of $a$ and $\Delta$ of maximum cardinality. By Lemma~\ref{claim:emptyChain} there is some chain of $a$ and $\Delta$, and by Definition~\ref{def:chain} a chain is a finite subset of elements; thus, $S^*$ is well defined. By Lemma~\ref{claim:qChain} it holds that $|S^*| \leq q(\eps)$; thus, by Step~\ref{step:q} and Step~\ref{step:Br} of Algorithm~\ref{alg:representative}, and \eqref{eq:Ex-set}, it holds that $B_{S^*}$ is computed by Algorithm \textsf{ExSet-MatroidIntersection}. Then, 
by Lemma~\ref{lem:maximum} there is a shift $b^* \in B_{S^*}$ to $a$ for $\Delta$. 
Note that $B_{S^*} \subseteq X$ by Step~\ref{step:RBS} of Algorithm~\ref{alg:representative}. Therefore, by Definition~\ref{def:shift} (shift) and Definition~\ref{def:r-set} (exchange set) it follows that $X$ is an exchange set of $I,\eps,\alpha$, and $r$. We use the following claim for the complexity analysis. 
	\begin{claim}
	\label{claim:r<4}
	$|X| \leq q(\eps)^{O(q(\eps)}$ and the running time of the algorithm is $q(\eps)^{O(q(\eps)} \cdot \textnormal{poly}(|I|)$.
\end{claim} 
\begin{claimproof}

 Recall that $\mathcal{S}$ is the set of all branches $S \subseteq {\ck}_r(\alpha)$ such that $\textsf{ExtendChain}(I,\eps,\alpha,r,S)$ is computed during the execution of $\textsf{ExSet-MatroidIntersection}(I,\eps,\alpha,r)$. For all $t \in \{0,1,\ldots, q(\eps)+1\}$ let 
		$L_t = \{S \in \mathcal{S}~|~|S| = t\}$ 
		be the set of branches $S$ containing $t$ elements computed in some iteration of Step~\ref{step:Br} of Algorithm~\ref{alg:representative} in the course of Algorithm \textsf{ExSet-MatroidIntersection}. 	Observe that the operations in each recursive call to Algorithm \textsf{ExtendChain} (i.e., without secondary recursive calls) can be computed in polynomial time in the instance size, as finding a minimum basis in Step~\ref{step:Br} can be done in linear time (see, e.g., \cite{cormen2022introduction}). Moreover, by Step~\ref{step:RBS} and Step~\ref{step:q}, the number of recursive calls to the algorithm is bounded by the number of branches reached by the algorithm. Note that the branches are of sizes between $0$ and $q(\eps)+1$ by \eqref{eq:Ex-set} and Step~\ref{step:q} of Algorithm~\ref{alg:representative}. 	In addition, each branch $S$ induces $|S|$ recursive calls to Algorithm~\ref{alg:representative}. Thus, the running time of the algorithm can be bounded by \begin{equation}
		\label{eq:1111}
		\begin{aligned}
		\sum_{t \in \{0,1,\ldots, q(\eps)+1\}} \sum_{S \in L_t} |B_S| \cdot \textnormal{poly}(|I|) \leq{} & 	\sum_{t \in \{0,1,\ldots, q(\eps)+1\}} |L_t| \cdot q(\eps) \cdot \textnormal{poly}(|I|)\\
		\leq{} & (q(\eps)+1) \cdot {q(\eps)}^{q(\eps)+1} \cdot q(\eps) \cdot  \textnormal{poly}(|I|)\\
		 = {}& q(\eps)^{O(q(\eps)} \cdot \textnormal{poly}(|I|). 
		\end{aligned}
	\end{equation} The first inequality holds since for every branch $S$ it holds that $|B_S| \leq q(\eps)$ by Step~\ref{step:Br} and Definition~\ref{def:matroids}. 	The second inequality holds since $\forall t \in [q(\eps)]: |L_{t+1}| \leq q(\eps) \cdot |L_t|$ by Step~\ref{step:RBS} of Algorithm~\ref{alg:representative}; thus, as $|L_0|  = 1$ by \eqref{eq:Ex-set}, by an inductive argument it holds that $\forall t \in [q(\eps)+1]: |L_t| \leq q(\eps)^{q(\eps)+1}$. 
Finally, for the size of $X$, the returned exchange set, we have	\begin{equation*}
		\label{eq:lastanalysis}
		\begin{aligned}
			|X| ={} & \left|\bigcup_{t \in \{0,1,\ldots, q(\eps)+1\}} \bigcup_{S \in L_t} B_S\right| \\
			\leq{} & \sum_{t \in \{0,1,\ldots, q(\eps)+1\}} \sum_{S \in L_t} |B_S| \\
			\leq{} & \sum_{t \in \{0,1,\ldots, q(\eps)\}} |L_t| \cdot q(\eps)  \\
			= {} & q(\eps)^{O(q(\eps))}.
		\end{aligned}
	\end{equation*} The first equality holds by \eqref{eq:Ex-set}, and Steps~\ref{step:q},~\ref{step:RBS} of Algorithm~\ref{alg:representative}. The second inequality holds by Step~\ref{step:Br}, since the computed bases $B_S$ are of sizes bounded by $q(\eps)$ by Definition~\ref{def:matroids}. The last inequality holds by symetric arguments to \eqref{eq:1111}. 
\end{claimproof}
The proof of Lemma~\ref{lem:mainMatroid} follows by the above.  \qed

\end{document}